\documentclass[english,a4paper,reqno,11pt]{amsart} 
\usepackage[utf8]{inputenc} 
\usepackage[T1]{fontenc} 
\usepackage{lmodern} 
\usepackage{amsmath}
\usepackage{amssymb} 
\usepackage{stmaryrd} 
\usepackage{mathtools} 
\usepackage{mathrsfs}  
\usepackage{graphicx} 
\usepackage{xcolor}   
\usepackage[makeroom]{cancel}
\usepackage[new]{old-arrows}
\usepackage[english]{babel} 
\usepackage{amsthm, thmtools} 
\usepackage{caption} 
\usepackage{subcaption}
\usepackage{varioref} 
\usepackage{microtype} 
\usepackage[shortlabels]{enumitem} 
\usepackage{xargs} 
\usepackage[colorinlistoftodos,
	prependcaption,
	textsize=scriptsize
	]
	{todonotes} 
\usepackage{csquotes} 

\usepackage{dsfont} 
\usepackage{bbm} 
\usepackage{tikz,tikz-cd} 
\tikzcdset{scale cd/.style={every label/.append style={scale=#1},
    cells={nodes={scale=#1}}}}
\usepackage{shuffle} 
\usepackage{blkarray} 

\usepackage[
	doi=false,
	isbn=false,
	backend=bibtex,
	style=alphabetic,
	maxbibnames=99,
	maxalphanames=99,
	giveninits=true,
	]
	{biblatex} 
\usepackage{bm}
\usepackage[hypertexnames=false]{hyperref} 

\newcommandx{\lz}[2][1=]{\todo[inline,linecolor=blue,backgroundcolor=blue!25,bordercolor=blue,#1, author = LORENZO]{#2}} 
\newcommandx{\jdj}[2][1=]{\todo[inline,linecolor=red,backgroundcolor=red!25,bordercolor=red,#1, author = JEAN-DAVID]{#2}} 

\usepackage{forest}

\forestset{
  decor/.style = {
    label/.expanded = {[inner sep = 0.1ex, font=\unexpanded{\footnotesize}]above:{$#1$}}
  },
  rdecor/.style = {
    label/.expanded ={[inner sep = 0.1ex, font=\unexpanded{\footnotesize}]below:{$#1$}},
  },
  root/.style = {minimum size = 0.5 ex, delay = {rdecor/.option = content, content =}},
  decorated/.style = {
    for tree = {
        circle, fill, inner sep = 0ex, minimum size = 1 ex,
      grow' = north, l = 0, l sep = 2 ex, s sep = 1.5 em,
      fit = tight, parent anchor = center, child anchor = center,
      edge = {thick},
    },
    for level>=1{
      delay = {decor/.option = content, content =}
    },
  },
  default preamble = {root, decorated},
  begin draw/.code={\begin{tikzpicture}[baseline={([yshift=-0.5ex]current bounding box.center)}]},
}

\setenumerate[0]{label = \textit{\arabic*.}, ref = \textit{\arabic*.}} 

\numberwithin{equation}{section} 

\newtheorem{prop}{Proposition}[section]
\newtheorem{definition}[prop]{Definition}
\newtheorem{definition-theorem}[prop]{Definition-Theorem}
\newtheorem{lemma}[prop]{Lemma}

\newtheorem{theorem}[prop]{Theorem}
\newtheorem{remark}[prop]{Remark}
\newtheorem{corollary}[prop]{Corollary}

\newtheorem{assumption}[prop]{Assumption}
\newtheorem{notations}[prop]{Notations}

\newcommand{\0}{{\bf 0}}
\newcommand{\1}{\mathsf{1}}
\newcommand{\m}{{\bf m}}

\newcommand{\n}{{\bf n}}

\newcommand{\p}{{\bf p}}
\newcommand{\q}{{\bf q}}
\newcommand{\e}{{\bf e}}
\newcommand{\z}{\mathsf{z}}
\newcommand{\A}{\mathcal{A}}
\newcommand{\B}{\mathcal{B}}
\newcommand{\Bbar}{\overline{\mathcal{B}}}

\newcommand{\D}{\mathcal{D}}
\newcommand{\E}{\mathcal{E}}
\newcommand{\F}{\mathcal{F}}

\newcommand{\K}{\mathbb{K}}
\renewcommand{\L}{\mathcal{L}}
\newcommand{\N}{\mathbb{N}}
\renewcommand{\P}{\mathcal{P}}
\newcommand{\R}{\mathbb{R}}
\newcommand{\sym}{\text{S}}
\newcommand{\tens}{\text{T}}
\newcommand{\env}{\mathcal{U}}
\newcommand{\envU}{\mathcal{U}_{[\cdot,\cdot]}}

\newcommand{\db}[1]{\llbracket {#1} \rrbracket}
\newcommand{\rhohat}{\hat{\rho}}

\newcommand{\id}{\textnormal{id}}
\newcommand{\End}{\textnormal{End}}
\newcommand{\Der}{\textnormal{Der}}

\newcommand{\conc}{\mathsf{conc}}

\newcommand{\ind}{\mathds{1}}

\newcommand{\tr}{\triangleright}
\newcommand{\la}{\left\langle}
\newcommand{\ra}{\right\rangle}

\newcommand{\otimesbold}{\bm{\otimes}}

\newcommand{\writefun}[5]{\ensuremath{\begin{array}[t]{lrcl}
#1 : & #2 & \longrightarrow & #3 \\
    & #4 & \longmapsto & #5 \end{array}}} 

\newcommand{\C}{\mathbb{C}}
\newcommand{\ass}{\mathsf{a}}

\newcommand{\Cop}{\Delta_\ast}
\newcommand{\M}{\mathcal{M}}

\renewcommand{\b}{\mathfrak{b}}

\renewcommand{\d}{\, {\rm d}}

\newcommand{\trbar}{\star}

\newcommand{\rhobar}{\overline{\rho}}

\title{Post-Lie algebras of derivations and regularity structures}
\author{Jean-David JACQUES}
\author{Lorenzo Zambotti}
\date{\today}
\begin{document}

\begin{abstract}
Given a commutative algebra $\A$, we exhibit a canonical structure of post-Lie algebra on the space $\A\otimes \Der(\A)$ where $\Der(\A)$ is the space of derivations on $\A$, in order to use the machinery given by Oudom-Guin (2008) and Ebrahimi-Fard--Lundervold--Munthe-Kaas (2015), and to define a Hopf algebra structure on the associated enveloping algebra with a natural action on $\A$. We apply these results to the setting of Linares-Otto-Tempelmayr (2023), giving a simpler and more efficient construction of their action and extending the recent work by Bruned-Katsetsiadis (2023).
This approach gives an optimal setting to perform explicit computations in the associated structure group.
\end{abstract}

\maketitle
\vspace{-1 cm}
\tableofcontents

\section{Introduction}

This paper concerns an algebraic structure recently unveiled in a remarkable series of papers \cite{OSSW,LOT,linares2022tree,LOTT} in the
context of regularity structures \cite{Hai14} and their applications to stochastic partial differential equations. In this paper
we explore this new structure and we propose a different construction.

There is a long history of applications of algebraic structures to numerical and, more recently, stochastic analysis.
In the context of Butcher series for the time-discretization of ordinary differential equations \cite{Butcher} and 
in the context of branched rough paths \cite{ramification} and their applications to stochastic differential equations, 
the main algebraic structure of interest is the Connes-Kreimer
Hopf algebra of rooted trees (or forests). In regularity structures, which are the natural evolution of branched rough paths
in the context of stochastic partial differential equations, the main algebraic objects are several Hopf algebras and comodules \cite{BHZ}
and pre-Lie algebras \cite{CL} on families of decorated rooted trees (or forests) \cite{BCCH}. 

The starting point of \cite{LOT} is the observation that Butcher series in all these contexts can be expressed as sums
over multi-indices rather than of trees: it is indeed possible to replace each (rooted) tree by its \emph{fertility}, namely the function
which, to each $k\in\N$, associates the number of vertices in the tree with exactly $k$ children. Surprisingly, many of the tree-based algebraic structures
have an analog in the multi-indices setting. The multi-indices algebraic structure is described by a representation in
an algebra of endomorphisms on a linear space; more precisely, in an algebra of \emph{derivations} on a space of formal power series.

The main aim of \cite{LOT} is then to give an abstract formulation of the composition product in their chosen space of derivations.
The \emph{parti pris} of \cite{LOT} is to construct such a product starting from a pre-Lie algebra \cite{CL} and using the Guin-Oudom procedure \cite{oudom2008lie}.
This approach works in the setting of the Grossman-Larson product \cite{grossmanlarson,Hoffman}, dual of the Butcher-Connes-Kreimer Hopf algebra, which is relevant for branched rough paths,
and the pre-Lie operation given by \cite{LOT} is the translation in the multi-indices setting of the {\it grafting} operation. 
However the authors of \cite{LOT} recognise that the operation on the space of derivations they define fails to satisfy the pre-Lie property in the SPDE-regularity
structures setting, and their construction 
becomes somewhat obscured by the technicalities needed to circumvent this problem. The recent paper \cite[\S 5]{bruned2022post} showed that 
the correct point of view in this setting is rather that of {\it post-Lie} algebras, a notion which generalises that of pre-Lie (see  Section 
\ref{sec:post} for all related definitions). Post-Lie algebras already play a role in so-called planarly branched rough paths \cite{CEFMM}.

In this paper we build on the intuition of \cite[\S 5]{bruned2022post} and we show that \cite{LOT} can be seen as a particular case of a more general construction: we consider a general commutative algebra $\A$ and we exhibit a canonical post-Lie algebra structure on the space $\A\otimes \Der(\A)$ where $\Der(\A)$ is the space of derivations on $\A$; the setting of \cite{LOT} can then be considered a sub-post-Lie algebra of $\A\otimes \Der(\A)$ for a certain choice of $\A$.

One of the main differences between our approach and that of \cite{LOT} is that we write a \emph{different} (albeit isomorphic) Hopf algebra. The point of view of \cite{LOT} is to construct a pre-Lie structure which generates a Lie-algebra on a specific space $L\subseteq\Der(\A)$ of derivations on a commutative algebra $\A$, where the Lie bracket is generated by the composition product: $\llbracket A,B \rrbracket:=A\circ B-B\circ A$. The Hopf algebra of \cite{LOT} is the universal enveloping algebra $\env_{\db{\cdot,\cdot}}(L)$ of this Lie algebra. 

In the post-Lie setting that we study, which extends the one introduced by \cite[\S 5]{bruned2022post}, there is a second and simpler Lie bracket denoted
by $[\cdot,\cdot]$. We use this bracket to construct a universal enveloping algebra $\envU(L)$ that becomes our main Hopf algebra. 
This Hopf algebra comes with a natural action on $\A$ which is the basis for the construction of the structure group of a regularity structure,
see Section \ref{sec:ansatz}. In this way we have a simpler abstract formulation of a non-commutative associative product $\trbar$ on $\envU(L)$, which makes
$\rhobar:(\envU(L),\trbar)\to(\End(\A),\circ)$ an algebra morphism. 
This framework seems to offer an optimal setting to perform computations related to this non-trivial product, see Sections \ref{subsec: basis env alg} and \ref{sec:expprod}.

Our construction uses some of the techniques developed by \cite{LOT} but rephrases them in a language closer to the original theory of regularity structures, which should be of interest for other readers; in several instances we borrow definitions and formulae from \cite{LOT}, reproving them in our way. 
We also mention that a second pre-Lie operation related to {\it insertion} at the level of trees and in cointeraction with the previous one
related to grafting \cite{ManchonSaidi,CMEF} is currently being investigated in the rough-paths setting \cite{Linares}, together with its extension to the SPDE-regularity structures case \cite{BrunedLinares}.
We also give a formula for the coproduct $\Delta_\trbar$ which is the dual of the Guin-Oudom product $\trbar$, see Proposition \ref{prop:trbar}.
This formula has recently been proved in the particular case of \cite{LOT} in \cite{GMZ} and \cite{BH24}.\\

The paper is organized as follows:

In Section \ref{sec:post} we recall generalities about pre-Lie algebras, post-Lie algebras and their universal enveloping algebras and we derive two minimal Assumptions \ref{assump: finiteness Lie bracket} and \ref{assump: finiteness tr} under which the product $\star$ on $\envU(L)$ can be dualised into a coproduct $\Delta_\star$, see Corollary \ref{lem:wd}.

In Section \ref{sec:derivations} we define a natural post-Lie structure on derivations on a commutative algebra $\A$, thus generalising a result of \cite{Burde} in the case of commuting derivations; we give explicit expressions for the associated Guin-Oudom product, see Proposition \ref{prop: extension product}, using the construction of \cite{ebrahimi2014lie} and the important representation $\rhobar:(\envU(L),\trbar)\to(\End(\A),\circ)$ of the universal enveloping algebra on $\A$.

In Section \ref{sec:LOT} we move to a particular case studied in \cite{LOT} and
we follow their definitions of a family of derivations on a fixed space of power sums. 
In Section \ref{sec:ansatz} we choose, similarly to \cite{LOT},
a stochastic PDE (see equation \eqref{eq:spde} below) and we construct the so-called  \emph{structure group} for this equation, which is the starting point of the regularity structures approach.

\medskip
{\bf Acknowledgements}. We thank Kurusch Ebrahim-Fard, Loïc Foissy, Pablo Linares and Markus Tempelmayr for very useful discussions on the topic of this paper.

\bigskip
\section{Post-Lie algebras and universal Lie enveloping algebra}\label{sec:post}
\medskip
\subsection{Lie algebras, post-Lie algebras}
A linear space $L$ endowed with a bilinear operation $L^{\otimes 2}\rightarrow L$, $a\otimes b\mapsto [a,b]$ is said to be a \textit{Lie algebra} if the following relations are satisfied for all $a,b,c\in L$:
\begin{enumerate}
    \item $[a,b]=-[b,a]$ \quad (anticommutativity);
    \item $[a,[b,c]]+[c,[a,b]]+[b,[c,a]]=0$ \quad (Jacobi relation).
\end{enumerate}

\begin{definition}\label{def:preLie}
A (left) \textbf{pre-Lie algebra} $(L,\tr)$ is the data of a vector space $L$, endowed with a bilinear operation $\tr:L\otimes L\rightarrow L$ which verifies the following relation for all $a,b,c\in L$:
\begin{equation}\label{eq: def pre-Lie relation}
    a\tr (b\tr c) - (a\tr b)\tr c = b\tr (a\tr c) - (b \tr a)\tr c.
\end{equation}
\end{definition}

Given a bilinear operation $\circ:L^{\otimes 2}\rightarrow L$ on a vector space $L$, its commutator bracket $[\cdot,\cdot]_\circ:L^{\otimes 2}\rightarrow L$ is defined as the commutator bracket
\[
[a,b]_\circ:=a\circ b-b\circ a,
\]
while its associator is a trilinear map $\ass_\circ:L^{\otimes 3}\rightarrow L$ defined as:
\begin{equation*}
    \ass_\circ(a,b,c):=a\circ(b\circ c) - (a\circ b)\circ c.
\end{equation*}
The associator measures the default of associativity: $\mathsf{\ass}_\circ(a,b,c)=0$ for all $a,b,c\in L$, if and only if $\circ$ is associative on $L$.
The pre-Lie relation \eqref{eq: def pre-Lie relation} is written in terms of the associator as
\begin{equation*}
    \ass_\tr(a,b,c)-\ass_\tr(b,a,c)=0.
\end{equation*}

%

\begin{definition}
A (left) \textbf{post-Lie algebra} $(L,\tr,[\cdot,\cdot])$ is a vector space $L$ endowed with two binary operations $\tr,[\cdot,\cdot]:L\otimes L\rightarrow L$ which satisfy for all $a,b,c\in L$ the following conditions:
\begin{enumerate}
    \item $[\cdot,\cdot]$ is a Lie bracket;
    \item $a\tr[b,c]=[a\tr b,c]+[b,a\tr c]$;
    \item $[a,b]\tr c=\ass_\tr(a,b,c)-\ass_\tr(b,a,c)$.
\end{enumerate}
\end{definition}

\begin{remark}
If $(L,\tr,[\cdot,\cdot])$ is a post-Lie algebra and $[\cdot,\cdot]\equiv 0$, then $(L,\tr)$ is a pre-Lie algebra. Vice versa, given $(L,\tr)$ a 
pre-Lie algebra, if we set $[\cdot,\cdot]\equiv 0$ then $(L,\tr,[\cdot,\cdot])$ is a post-Lie algebra.
\end{remark}

In a pre-Lie algebra $(L,\tr)$, the commutator given by: $$[a,b]_\tr:=a\tr b - b\tr a$$ verifies the Jacobi identity and thus is a Lie bracket. On the other hand, in a post-Lie algebra $(L,\tr,[\cdot,\cdot])$ the commutator $[a,b]_\tr$ is not in general a Lie bracket; however, we have the following
\begin{prop}[\cite{ebrahimi2014lie}] \label{pr:ebrahimi2014lie}
Let $(L,\tr,[\cdot,\cdot])$ be a post-Lie algebra. The bilinear operation $\llbracket \cdot,\cdot \rrbracket:L\otimes L\rightarrow L$ defined for all $a,b\in L$ by:
\begin{equation} \label{eq: composition Lie bracket}
    \llbracket a,b \rrbracket: = a\tr b-b\tr a +[a,b]
\end{equation}
is a Lie bracket, that we will call here the \textbf{composition Lie bracket}.
\end{prop}
\medskip

\subsection{The Lie enveloping algebra}\label{subsec: Env. alg.}

Given a Lie algebra $(L,[\cdot,\cdot])$, we denote by $\tens(L)=\bigoplus_{k\geq 0} L^{\otimes k}$  the tensor algebra over $L$ (with the convention that $L^0=\R\{\ind\}$), whose elements are, given a basis $\B_L$ of $L$, linear combinations of (non-commutative) monomials often called \textit{words} $a_1\otimes \cdots \otimes a_n$ (also noted simply
$a_1\cdots a_n$ if no confusion arises) for $(a_1,\ldots,a_n)\in (\B_L)^{n}$.

The \textit{Lie enveloping algebra} of a Lie algebra $(L,[\cdot,\cdot])$, denoted $\env_{[\cdot,\cdot]}(L)$, is defined as the tensor algebra $\tens(L)=\bigoplus_{k\geq 0} L^{\otimes k}$ over $L$ quotiented by the two-sided ideal $\mathfrak{c}$ generated by $\{a\otimes b - b\otimes a- [a,b]:\,a,b\in L\}$:
$$\env_{[\cdot,\cdot]}(L):=\tens(L)/ \mathfrak{c}.$$

The vectors of $\env_{[\cdot,\cdot]}(L)$ are by definition equivalence classes on $\tens(L)$.
Since no confusion can occur, we will adopt the same notation $a_1 \cdots a_n$ for the equivalence class in $\env_{[\cdot,\cdot]}(L)$ as for its representative in $\tens(L)$. We have a canonical injection $\R\ni t\mapsto
t\ind\in\env_{[\cdot,\cdot]}(L)$ and the counit map $\env_{[\cdot,\cdot]}(L)\ni x\mapsto \varepsilon(x)\in\R$ where 
$x-\varepsilon(x){\ind}\in \bigoplus_{k\geq 1} L^{\otimes k}/{\mathfrak c}$.

A natural filtration can be given on the enveloping algebra: denoting $\tens^{(n)}(L):=\bigoplus_{k=0}^n L^{\otimes k}$ for $n\in\N$ and $\mathfrak{c}^{(n)}:=\mathfrak{c}\cap \tens^{(n)}(L)$,  one has the following sequence of inclusions:
\begin{equation}\label{eq:filtration}
\env^{(0)}_{[\cdot,\cdot]}(L) \subset \env^{(1)}_{[\cdot,\cdot]}(L)\subset \env^{(2)}_{[\cdot,\cdot]}(L)\subset \dots \subset \env_{[\cdot,\cdot]}(L),
\end{equation}
where $\env^{(0)}_{[\cdot,\cdot]}(L)=\R\ind$, $\env^{(1)}_{[\cdot,\cdot]}(L)=\R\ind\oplus L$ and $\env^{(n)}_{[\cdot,\cdot]}(L)=\tens^{(n)}(L)/\mathfrak{c}^{(n)}$ for all $n\geq 2$. Then obviously:
$$\env_{[\cdot,\cdot]}(L)=\bigcup_{n=1}^\infty \env^{(n)}_{[\cdot,\cdot]}(L).$$

In the rest of the paper, we will consider $L$ as a subspace of $\env_{[\cdot,\cdot]}(L)$ by the composition of the canonical injection into the tensor algebra composed with the projection:
$$L\hookrightarrow \tens(L) \twoheadrightarrow \env_{[\cdot,\cdot]}(L).$$

The space $\env_{[\cdot,\cdot]}(L)$ inherits from $\tens(L)$ the associative algebra structure $({\conc},{\mathds 1})$, where $\mathsf{conc}:\env_{[\cdot,\cdot]}(L)\otimes \env_{[\cdot,\cdot]}(L)\to \env_{[\cdot,\cdot]}(L)$ is the concatenation product:
$$\mathsf{conc}:a_1\cdots a_n \otimes b_1\cdots b_m \mapsto a_1\cdots a_n  b_1\cdots b_m.$$
Then $\env_{[\cdot,\cdot]}(L)$ endowed with the concatenation product $\mathsf{conc}$ is an algebra with unit $\ind$. 

In the particular case of the trivial Lie algebra $(L,[\cdot,\cdot])$ with null Lie bracket $[\cdot,\cdot]\equiv 0$, the algebra $(\env_{[\cdot,\cdot]}(L),\mathsf{conc},\ind)$ is the \textit{symmetric tensor algebra}, denoted $(\sym(L),\mathsf{conc},\ind)$. It is a commutative algebra which is isomorphic to the polynomial algebra $\R[\B_L]$ once a basis $\B_L$ of $L$ has been fixed.

If $[\cdot,\cdot]$ is non-trivial, the order of the letters in the monomials of $\env_{[\cdot,\cdot]}$ matters and the following famous theorem permits to exhibit a basis for $\env_{[\cdot,\cdot]}(L)$.

\begin{theorem}[Poincaré-Birkhoff-Witt]\label{theo: PBW}
    Given a basis $\B_L$ of $L$ and a total order $\leq$ on it, a basis $\B_{\envU(L)}=\B_{\envU(L)}^\leq$ of $\env_{[\cdot,\cdot]}(L)$ is given by 
    \begin{equation}\label{eq:basis}
\begin{split}
\B_{\envU(L)}:=&\{\mathds 1\}\sqcup\Bigg\{\frac1{m_1!\cdots m_k!}\,x_1^{m_1}\cdots x_k^{m_k}: \   k, m_1,\ldots,m_k\geq 1, \\
& \qquad \qquad x_1<\ldots<x_k, \ x_i\in\B_L
\Bigg\}.
\end{split}
\end{equation}
\end{theorem}

The enveloping algebra gives a functor $L\mapsto \env_{[\cdot,\cdot]}(L)$ from the category of Lie algebras to the category of associative algebras 
which satisfies the following universal property:
\begin{theorem}[\textbf{Universal property}]\label{theo: universal property env alg}
    Given a Lie algebra $(L,[\cdot,\cdot])$, an associative algebra $(A,\circ)$ and a Lie algebra morphism $\varphi:(L,[\cdot,\cdot])\to (A,[\cdot,\cdot]_\circ)$,
namely such that $\varphi([a,b])=[\varphi(a),\varphi(b)]_\circ$ for all $a,b\in L$,     there exists a unique algebra morphism $\bar\varphi:(\env_{[\cdot,\cdot]}(L),\conc) \to (A,\circ)$ such that $\bar\varphi(a)=\varphi(a)$ for all $a\in L$.
\end{theorem}
\medskip

\subsection{The coshuffle coproduct and its dual product}\label{sec:ast}

It is a well known fact that there exists a unique coproduct $\Cop:\env_{[\cdot,\cdot]}(L)\to \env_{[\cdot,\cdot]}(L)\otimes \env_{[\cdot,\cdot]}(L)$, which turns $\env_{[\cdot,\cdot]}(L)$ into a bialgebra $(\env_{[\cdot,\cdot]}(L),\conc,\Cop,\ind,\varepsilon)$ for which the Lie algebra of primitive elements is $L$, in other terms:
$$\Cop(a)=a\otimes \mathds 1 + \mathds 1\otimes a\qquad \text{for all}~a\in L$$
and the counit map $\varepsilon: \env_{[\cdot,\cdot]}(L)\to\R$ is the linear map given by $\varepsilon(\ind)=1$ and  $Ker(\varepsilon)=\bigoplus_{k\geq 1}L^{\otimes k}/\mathfrak{c}$.

The existence and uniqueness of $\Cop$ is guaranteed by the universal property \ref{theo: universal property env alg} owing to the fact that 
$[\Cop(a),\Cop(b)]_\conc=\Cop[a,b]_\conc$ for all $a,b\in L$, which indicates that $\Cop:L\to \env_{[\cdot,\cdot]}(L)\otimes \env_{[\cdot,\cdot]}(L)$ is a Lie algebra morphism. Coassociativity and cocommutativity are easily proved on $L$ and extended by multiplicativity on $\env_{[\cdot,\cdot]}(L)$, as well as the counit property, see \cite[\S II.1.4]{nicolas1989lie}.

On (equivalence classes of) words we have
\begin{equation}\label{eq:coshuffle}
    \Cop (a_1 \cdots a_n)=\Cop (a_1)\cdots\Cop (a_n)= \sum_{I\subseteq\{1,\ldots,n\}} a_I \otimes a_{\{1,\ldots,n\}\setminus I},
\end{equation}
where we denote:
\begin{equation}\label{eq:notationaI}
a_\emptyset:=\mathds{1}, \qquad
a_I:=a_{i_1} \ldots a_{i_p}, \qquad I=\{i_1,\ldots,i_p\}, \qquad i_1<\cdots<i_p.
\end{equation}

On the basis $\B_{\envU(L)}$ in \eqref{eq:basis} the coshuffle coproduct has a very convenient form
\begin{equation}\label{eq:cop1}
\begin{split}
\Cop \prod_{i=1}^k \frac{x_i^{m_i}}{m_i!} &= \prod_{i=1}^k \Cop  \frac{x_i^{m_i}}{m_i!}
= \prod_{i=1}^k \sum_{\ell=0}^{m_i} \frac{x_i^{\ell}}{\ell!}\otimes\frac{x_i^{m_i-\ell}}{(m_i-\ell)!}
\\ & = \sum_{0\le\ell_i\le m_i} \left(\prod_{i=1}^k\frac{x_i^{\ell_i}}{\ell_i!}\right)\otimes
\left(\prod_{i=1}^k\frac{x_i^{m_i-\ell_i}}{(m_i-\ell_i)!}\right),
\end{split}
\end{equation}
which is the reason for the normalisation chosen in \eqref{eq:basis}. We often use Sweedler's notation 
\begin{equation}\label{eq:Sweedler}
\Cop u = \sum_{(u)} u^{(1)}\otimes u^{(2)}.
\end{equation}
\medskip

\subsection{Hopf algebra structure on the post-Lie enveloping algebra}\label{sec:hopfa}
In the case of a pre-Lie algebra $(L,\tr)$, Guin-Oudom \cite{oudom2008lie} developed a procedure in order to extend the pre-Lie product to the symmetric tensor algebra $\sym(L)$, and defined a product $\trbar$ which turns $(\sym(L),\trbar,\Cop)$ into an associative and cocommutative Hopf algebra. The space $L$, considered as a subspace of $\sym(L)$, turns out to be the Lie algebra of primitive elements for the bracket $\llbracket\cdot,\cdot\rrbracket$ in \eqref{pr:ebrahimi2014lie}. 
The Cartier-Milnor-Moore theorem for filtered Hopf algebra (see \cite{milnor1965structure} 
or \cite[Theorem 1, \S II.6]{nicolas1989lie} for the filtered bialgebra version) applies and gives an isomorphism of Hopf algebras between $(\sym(L),\trbar,\Cop)$ and $(\env_{\llbracket\cdot,\cdot\rrbracket}(L),\mathsf{conc}, \Cop)$. 

Later in \cite{ebrahimi2014lie}, the authors showed that the machinery developped in \cite{oudom2008lie} in the case of pre-Lie algebras can be applied to the more general case of post-Lie algebras, giving an extension of the post-Lie product $\tr$ to $\env_{[\cdot,\cdot]}(L)$ and an associative product $\trbar$ which turns $(\env_{[\cdot,\cdot]}(L),\trbar,\Cop)$ into an associative Hopf algebra. The Milnor-Moore theorem applies again and gives an isomorphism of Hopf algebras between $(\envU(L),\trbar,\Cop)$ and $(\env_{[\cdot,\cdot]_\trbar}(L),\mathsf{conc}, \Cop)$, as we will see below. If the bracket $[\cdot,\cdot]$ is null, the concatenation product of $\env_{[\cdot,\cdot]}(L)$ is commutative, and the space is equal to $\sym(L)$, which gives back the case of pre-Lie algebras. We refer to the monograph \cite{Patras} for the details of the theory of Hopf algebras.

\medskip
First let us recall the extension of the product $\tr$ to all $u,v\in \env_{[\cdot,\cdot]}(L)$, see Proposition 3.1 in \cite{ebrahimi2014lie}.

\begin{prop}
\label{prop: extension post-lie product}
Let $(L,\tr,[\cdot,\cdot])$ be a (left) post-Lie algebra. There exists a unique extension of the product $\tr$ to $\env_{[\cdot,\cdot]}(L)$ which verifies for all $a\in L$ and $u,v,w\in \env_{[\cdot,\cdot]}(L)$:
\begin{enumerate}
    \item $\ind \tr u = u$, $u\tr \ind=\varepsilon(u)$;
    \item $av\tr w=a\tr(v\tr w)-(a\tr v)\tr w$;
    \item $u\tr (v w)=\sum_{(u)} (u^{(1)}\tr v)(u^{(2)}\tr w)$;
\end{enumerate}
with Sweedler's notation \eqref{eq:Sweedler} for the coproduct \eqref{eq:coshuffle}.
\end{prop}
By definition, $L$ is the space of primitive elements in $(\env_{[\cdot,\cdot]}(L),\Cop)$, which means that for all $a\in L$: $\Cop(a)=a\otimes \mathds{1}+\mathds{1}\otimes a$. 
By property (3) in Proposition \ref{prop: extension post-lie product} and by \eqref{eq:coshuffle}, 
for all $a\in L$ and $b_1,\ldots,b_n\in\env_{[\cdot,\cdot]}(L)$ we have
\begin{equation}\label{eq:derivation}
a\tr(b_1\cdots b_n) = \sum_{i=1}^n b_1 \cdots (a\tr b_i) \cdots b_n.
\end{equation}
More generally, for all $a_1,\ldots,a_m\in L$ and $b_1,\ldots, b_n\in \env_{[\cdot,\cdot]}(L)$ we have by \eqref{eq:coshuffle}
\begin{equation}\label{eq: grafting word on word}
    a_1\cdots a_m \tr b_1\cdots b_n= \sum_{I_1\sqcup\cdots\sqcup I_n=\{1,\ldots,m\}} (a_{I_1}\tr b_1) \cdots (a_{I_n}\tr b_n)
\end{equation}
where we use the notation \eqref{eq:notationaI}.

Proposition \ref{prop: associative product} below appears in \cite[Proposition 3.3]{ebrahimi2014lie}, which extends the Guin-Oudom approach \cite{oudom2008lie}, originally used in the case of a pre-Lie algebra, to the case of a post-Lie algebra:

\begin{prop} \label{prop: associative product} 
Let $(L,\tr,[\cdot,\cdot])$ be a post-Lie algebra. The product $$\trbar:\env_{[\cdot,\cdot]}(L)\otimes \env_{[\cdot,\cdot]}(L)\to \env_{[\cdot,\cdot]}(L)$$ defined by $\star:=\conc\circ(\id\otimes \tr)\circ(\Delta_\ast\otimes \id)$ 
is associative and $(\env_{[\cdot,\cdot]}(L),\trbar,\Cop,\ind,\varepsilon)$ is a connected filtered Hopf algebra for the filtration given by \eqref{eq:filtration}, which antipode $S$ is given by $S(\ind)=\ind$ and the following formula on $\ker(\varepsilon)$:
\[
    S=-\id+\sum_{n\geq 1} (-1)^n \star^n\circ (\Delta_\ast')^n,
\]
where $\Delta_\ast':= \Delta_\ast-\ind\otimes \id -\id\otimes \ind$ denotes the reduced coproduct.
\end{prop}

First of all, we show that $\star$ respects the filtration \eqref{eq:filtration}, indeed for all $n,m\geq 0$:
\begin{align*}
    \envU^{(n)}\otimes \envU^{(m)}&\xrightarrow[]{\Delta_\ast\otimes \id} \bigoplus_{p,q\geq0,~p+q=n}\envU^{(p)}\otimes \envU^{(q)}\otimes \envU^{(m)}\\
    &\xrightarrow[]{\id\otimes \tr}\bigoplus_{p,q\geq0,~p+q=n}\envU^{(p)}\otimes \envU^{(q+m)}\\
    &\xrightarrow[]{\conc}\bigoplus_{p,q\geq0,~p+q=n}\envU^{(p+q+m)}=\envU^{(n+m)}.
\end{align*}
 We recall that every connected filtered bialgebra is a filtered Hopf algebra, see \cite[Corollary 5]{MHandbook} or
\cite[Theorem 3.4]{GG19}, also for the antipode formula. Using Sweedler's notation \eqref{eq:Sweedler} for the coshuffle coproduct $\Delta_\ast$ and recalling \eqref{eq:coshuffle}, we can write the following formula for all $u,v\in \env_{[\cdot,\cdot]}(L)$:
\begin{equation}\label{eq: post-Lie associative product}
    u\trbar v = \sum_{(u)} u^{(1)} (u^{(2)}\tr v).
\end{equation}
Since $L$ is the space of primitive elements in $(\env_{[\cdot,\cdot]}(L),\Cop)$, by definition of $\trbar$, for all $a,b\in L$ one has: 
\begin{equation}\label{eq:atrb}
    a \trbar  b=a\tr b + ab.
\end{equation}
The space $L$, considered as a subspace of $\env_{[\cdot,\cdot]}(L)$, is stable by the commutator \[[a,b]_\trbar:=a\trbar  b-b\trbar  a.\] By associativity of $\trbar$, $[\cdot,\cdot]_\trbar$ is thus a Lie bracket on $L$, and for all $a,b\in L\subset\env_{[\cdot,\cdot]}(L)$ we have
\begin{align*}
    [a,b]_\trbar &= a \trbar  b - b\trbar  a\\
    &= a\tr b - b\tr a + ab - ba\\
    &= a\tr b-b\tr a  + [a,b]=\llbracket a,b \rrbracket,
\end{align*}
where $\llbracket a,b \rrbracket$ is defined in \eqref{eq: composition Lie bracket}.
We thus deduce the equality between brackets for all $a,b\in L\subset\env_{[\cdot,\cdot]}(L)$
\begin{equation}\label{eq:=}
    [a,b]_\trbar=\llbracket a,b \rrbracket.
\end{equation}
Remark that the bracket $\llbracket \cdot,\cdot \rrbracket$ is defined intrinsically on the space $L$, while $[\cdot,\cdot]_\trbar$ is defined extrinsecally since $\trbar$ is a binary operation of $\env_{[\cdot,\cdot]}(L)$.

The Cartier-Milnor-Moore theorem for filtered algebras (see \cite[Theorem 1, \S II.6]{nicolas1989lie}) and the equality between brackets \eqref{eq:=} imply 
that $\left(\env_{\llbracket \cdot,\cdot \rrbracket}(L),\mathsf{conc},\Cop\right)$ and $\left(\env_{[\cdot,\cdot]}(L),\trbar,\Cop\right)$ are isomorphic as Hopf algebras.
In fact the isomorphism can be made very explicit:
\begin{theorem}\label{theo: isomorphism between enveloping algebras} The linear map $\Phi:\env_{\llbracket \cdot,\cdot \rrbracket}(L)\rightarrow \env_{[\cdot,\cdot]}(L)$ defined by:
\begin{equation*}
    \Phi(a_1 \cdots a_n):=a_1 \trbar \cdots \trbar a_n, \qquad a_1,\ldots,a_n\in L,
\end{equation*}
is an isomorphism of Hopf algebras $\left(\env_{\llbracket \cdot,\cdot \rrbracket}(L),\mathsf{conc},\Cop\right)\rightarrow\left(\env_{[\cdot,\cdot]}(L),\trbar,\Cop\right)$.
\end{theorem}

\begin{proof} 
This is \cite[Theorem 3.14]{oudom2008lie} in the pre-Lie case, which has been extended to the post-Lie case in \cite[Theorem 3.4]{ebrahimi2014lie}, see also
\cite[Proposition 4]{Foissy} and \cite[Theorem 10]{Mencattini}.
\end{proof}

We note the following extension of \eqref{eq:atrb}: for $a_0,a_1,\ldots,a_n\in L$ we have
\begin{equation}\label{eq:atrbbb}
\begin{split}
a_0\trbar (a_1\cdots a_n) & = a_0\tr (a_1\cdots a_n)+ a_0a_1\cdots a_n
\\ & = \sum_{i=1}^n a_1 \cdots (a_0\tr a_i) \cdots a_n+ a_0a_1\cdots a_n,
\end{split}
\end{equation}
where we have used \eqref{eq:derivation} in the last equality.

\medskip
\subsection{The dual structure}\label{sec:dual}

Recalling the basis $\B_{\envU(L)}$ of $\envU(L)$ from \eqref{eq:basis} given by the PBW Theorem \ref{theo: PBW}, 
we introduce now a second basis on $\envU(L)$ given by 
    \begin{equation}\label{eq:basis*}
\begin{split}
{\overline\B}_{\envU(L)}:=&\{\mathds 1\}\sqcup\{x_1^{m_1}\cdots x_k^{m_k}: \ k, m_1,\ldots,m_k\geq 1, \\
& \qquad  \qquad x_1<\cdots<x_k, \ x_i\in\B_L\}.
\end{split}
\end{equation}
We have a map $T:\B_{\envU(L)}
\to{\overline\B}_{\envU(L)}$ given by $T(\mathds 1)=\mathds 1$ and
\begin{equation}\label{eq:T}
T\left(\frac1{m_1!\cdots m_k!}x_1^{m_1}\cdots x_k^{m_k}\right)=x_1^{m_1}\cdots x_k^{m_k},
\end{equation}
whish has a unique linear extension $T:\envU(L)\to\envU(L)$.
Then we introduce the pairing on $\envU(L)\otimes\envU(L)$ given by the bilinear extension of 
\begin{equation}\label{eq:dualityL}
\B_{\envU(L)}\times{\overline\B}_{\envU(L)}\ni(u,v)\mapsto \la u,v\ra :=\ind_{(Tu=v)}.
\end{equation}
Then we can define an associative and commutative product $\ast$ on $\envU(L)$:
\begin{equation}\label{eq:*}
\begin{split}
u*v &:=\sum_{w\in\B_{\envU(L)}}\la \Cop w,u\otimes v\ra \,Tw
\\ & = 
\sum_{w\in\B_{\envU(L)}}\sum_{(w)}
\la w^{(1)},u\ra \,\la w^{(2)},v\ra \,Tw,
\end{split}
\end{equation}
 which is dual to the coproduct $\Cop$ in the sense that for all $u,v,w\in \envU(L)$, we have
\begin{equation*}
\la w,u*v\ra = \sum_{(w)} \la w^{(1)},u\ra \,\la w^{(2)},v\ra =\la \Cop w,u\otimes v\ra,
\end{equation*}
where we use Sweedler's notation \eqref{eq:Sweedler} for the coproduct \eqref{eq:coshuffle}.

The multiplication table of $\ast$ on ${\overline\B}_{\envU(L)}$  is given as follows:
\begin{equation}\label{coro: commutative product env alg}
\prod_{i=1}^k x_i^{\alpha_i}\ast\prod_{i=1}^k x_i^{\beta_i}=\prod_{i=1}^k x_i^{\alpha_i+\beta_i}
\end{equation}
for all $x_1<\ldots<x_k$ with $x_i\in\B_L$ and $\alpha_i,\beta_i\in\N$.
Therefore, we obtain from \eqref{eq:cop1}-\eqref{coro: commutative product env alg} 
the following relation between $\Cop$ and the product $\ast$ in \eqref{coro: commutative product env alg}
\begin{equation}\label{eq:copast}
\Cop u = \sum_{\substack{u_1,u_2\in\B_{\envU(L)}\\ (Tu_1)\ast (Tu_2) = Tu}} u_1\otimes u_2, \qquad u\in\B_{\envU(L)}.
\end{equation}
We stress that we use $\B_{\envU(L)}$ in \eqref{eq:basis} as a basis for $(\envU(L),\trbar)$ and 
${\overline\B}_{\envU(L)}$ in \eqref{eq:basis*} as a basis for $(\envU(L),\ast)$.
\begin{remark}
The choice of the basis $\B_{\envU(L)}$ in \eqref{eq:basis} and of the duality \eqref{eq:dualityL} may look unnatural, with respect to
the basis ${\overline\B}_{\envU(L)}$ in \eqref{eq:basis*}. One one hand, the basis $\B_{\envU(L)}$ in \eqref{eq:basis}
gives a particularly simple form to the coshuffle coproduct $\Cop$, see \eqref{eq:cop1}-\eqref{eq:copast}. On the other hand, the basis 
${\overline\B}_{\envU(L)}$ in \eqref{eq:basis*} and the duality \eqref{eq:dualityL} give the multiplication table \eqref{coro: commutative product env alg} for $\ast$, which corresponds to the polynomial product in the symmetric algebra over $L$ in the pre-Lie case, for example in the
Butcher-Connes-Kreimer Hopf algebra.
\end{remark}

A coagebra structure like $(\envU(L),\Cop,\varepsilon)$ endowed with a pairing $\la \cdot,\cdot \ra$ like in \eqref{eq:dualityL}, and admitting dual bases like $\B_{\envU(L)}$ and ${\overline\B}_{\envU(L)}$, can always be dualised by the formula \eqref{eq:*} into an algebra structure $(\envU(L),\ast,\mathds 1)$.

However, an algebra structure like for example $(\envU(L),\trbar,\ind)$ can not always be dualised into a coalgebra structure.
Therefore in order to define $\Delta_{\trbar}:\envU(L)\rightarrow \envU(L)\otimesbold \envU(L)$ such that
$$\la u_1\trbar  u_2,v \ra =\la u_1 \otimes u_2,\Delta_{\trbar}v \ra, \qquad \forall v,u_1,u_2\in\envU(L),$$
we need to make the following finiteness assumption on $L$:

Under the finiteness assumptions \ref{assump: finiteness Lie bracket} and \ref{assump: finiteness tr} that we are going to introduce, Corollary \ref{lem:wd} below will ensure that for all $v\in \envU(L)$ the following sum is well-defined, proving the existence of the coproduct dual to the product $\star$:
\[
    \Delta_\star v:=\sum_{u_1,u_2\in\B_{\envU(L)}}\la u_1\trbar  u_2,v \ra (Tu_1)\otimes (Tu_2)
\]
We define the length of $w=w_1\cdots w_n\in \Bbar_{\envU(L)}$ with $w_i\in\B_L$ by $\ell(w):=n$ (and $\ell(\mathds 1):=0$).
\begin{lemma}\label{lem: permutations}
    For $w\in\Bbar_{\envU(L)}$ of length $\ell(w)\geq 1$, and $a_1,\ldots,a_{\ell(w)}\in\B_L$, we have:
    $$\la a_{1}\ldots a_{\ell(w)},w\ra=\la a_{\sigma(1)}\ldots a_{\sigma(\ell(w))},w\ra$$
    for every permutation $\sigma$ of $\{1,\ldots,\ell(w)\}$.
\end{lemma}
\begin{proof} If $\ell(w)=1$, then there is nothing to prove. Let suppose that $\ell(w)\geq 2$. By definition of the Lie enveloping algebra $\envU(L)$, for all $a_1,\ldots,a_{\ell(w)}\in\B_L$:
    $$a_1\ldots a_i a_{i+1}\ldots a_{\ell(w)}=a_1\ldots a_{i+1}a_i\ldots a_{\ell(w)} + a_1\ldots [a_i,a_{i+1}]\ldots a_{\ell(w)}$$
    and $a_1\ldots [a_i,a_{i+1}]\ldots a_{\ell(w)}\in \envU^{(\ell(w)-1)}(L)$, see \eqref{eq:filtration}. By definition of the pairing $\la\cdot,\cdot\ra$ (see \eqref{eq:dualityL}), one has that:
    $$\la a_1\ldots [a_i,a_{i+1}]\ldots a_{\ell(w)},w\ra =0$$
    and thus we obtain that for all $i\in\{1,\ldots,\ell(w)-1\}$:
    $$\la a_1\ldots a_i a_{i+1}\ldots a_{\ell(w)},w\ra = \la a_1\ldots a_{i+1}a_i\ldots a_{\ell(w)},w\ra.$$
    We conclude the proof by recalling that the set of adjacent transpositions $\{(i,i+1):i=1,\ldots,n-1\}$ generates the symmetric group on $n$ elements.
\end{proof}

\begin{assumption}\label{assump: finiteness Lie bracket}
    For all $c\in\B_L$ the following set is finite:
    \[\Big\{(a,b)\in \B_L\times\B_L: \la [a,b],c \ra\neq 0\Big\}.\] 
\end{assumption}

\begin{lemma}\label{lem: finiteness Lie bracket}
     Under Assumption \ref{assump: finiteness Lie bracket}, for all $w\in\Bbar_{\envU(L)}$ of length $\ell(w)\geq 1$ and for all $m\in\N$, the following set is finite:
     \[\Big\{(a_1,\ldots, a_m)\in(\B_L)^m,~\la a_1\cdots a_m,w\ra \neq 0\Big\}.\]
\end{lemma}
\begin{proof}
    Fix $w\in\Bbar_{\envU(L)}$ of length $\ell(w)\geq 1$. For all $m\in\N$, let us denote by $\P(m)$ the assertion:
    \begin{center}
    $\P(m)$: "the set $\{(a_1,\ldots, a_m)\in(\B_L)^m,~\la a_1\cdots a_m,w\ra \neq 0\}$ is finite."
    \end{center}
    \begin{enumerate}[leftmargin=*]
        \item If $m<\ell(w)$, then $\la a_1\cdots a_m,w\ra=0$ for all $(a_1,\ldots, a_m)\in(\B_L)^m$, because $a_1\cdots a_m\in \envU^{(m)}(L)\subset \envU^{(\ell(w)-1)}(L)$.
        \item If $m=\ell(w)$, then from Lemma \ref{lem: permutations} $\la a_1\cdots a_m,w\ra \neq 0$ if and only if there exists a permutation $\sigma$ of $\{1,\ldots,m\}$ such that $w=a_{\sigma(1)}\cdots a_{\sigma(\ell(w))}$ and the number of such permutation is at most $\ell(w)!$.
        \item Now suppose that the finiteness property is proved up to $m-1\geq \ell(w)$.
        Take $a_1\ldots,a_{m}\in(\B_L)^{m}$, and consider a permutation $\sigma$ of $\{1,\ldots,m\}$ such that $a_{\sigma(1)}\leq \ldots \leq a_{\sigma(m)}$. Let us write 
        $\sigma$ as a composition of adjacent transpositions: $$\sigma=(i_1,i_1+1)\circ \cdots\circ (i_k,i_k+1),$$
        with $i_1,\ldots,i_k\in\{1,\ldots,m\}$.
        Consider the family $\{\sigma_0,\ldots,\sigma_k\}$ of permutations of $\{1,\ldots,m\}$ defined by: 
        $$\sigma_0=\id,\qquad \forall \ell \in\{0,\ldots,k-1\}:~ \sigma_{\ell}=(i_1,i_1+1)\circ \cdots\circ (i_\ell,i_\ell+1).$$
        We have for all $\ell\in\{1,\ldots,k\}$ that $\sigma_\ell=\sigma_{\ell-1}\circ(i_\ell,i_\ell+1)$ and therefore
	\[
        \begin{split}
          &  a_{\sigma_\ell(1)}\cdots a_{\sigma_\ell(m)}= a_{\sigma_{\ell-1}(1)}\cdots a_{\sigma_{\ell-1}(i_\ell+1)}a_{\sigma_{\ell-1}(i_\ell)} \cdots a_{\sigma_\ell(m)}
            \\&=a_{\sigma_{\ell-1}(1)}\cdots a_{\sigma_{\ell-1}(m)} -  a_{\sigma_{\ell-1}(1)}\cdots [a_{\sigma_{\ell-1}(i_\ell)},a_{\sigma_{\ell-1}(i_\ell+1)}] \cdots a_{\sigma_{\ell-1}(m)}.
        \end{split}
        \]
        Iterating, we obtain since $\sigma_0=\id$ and $\sigma_k=\sigma$:
	\[
        \begin{split}
           a_{\sigma(1)}\cdots a_{\sigma(m)} = & \, a_1\cdots a_{m}
            \\& - \sum_{\ell=1}^k  a_{\sigma_{\ell-1}(1)}\cdots [a_{\sigma_{\ell-1}(i_\ell)},a_{\sigma_{\ell-1}(i_\ell+1)}] \cdots a_{\sigma_{\ell-1}(m)}.
        \end{split}
        \]
        As $m>\ell(w)$, we have by definition of the pairing \eqref{eq:dualityL}, that: $$\la a_{\sigma(1)}\cdots a_{\sigma(m)},w\ra=0.$$ Therefore, by linearity:
        $$\la a_1\cdots a_{m},w\ra = \sum_{\ell=1}^{k} \la a_{\sigma_{\ell-1}(1)}\cdots [a_{\sigma_{\ell-1}(i_\ell)},a_{\sigma_{\ell-1}(i_\ell+1)}] \cdots a_{\sigma_{\ell-1}(m)} ,w \ra.
$$
        Then, $\la a_1\cdots a_{m},w\ra \neq 0$ implies that:
\[
            \exists \ell\in\{0,\ldots,k-1\},~\la a_{\sigma_\ell(1)}\cdots [a_{\sigma_\ell(i_l)},a_{\sigma_\ell(i_l+1)}] \cdots a_{\sigma_\ell(m)}, w\ra\neq 0.
\]
       For such an $\ell\in\{0,\ldots,k-1\}$, we write 
       \[
       [a_{\sigma_\ell(i_l)},a_{\sigma_\ell(i_l+1)}]=\sum_{d\in\B_L} \la [a_{\sigma_\ell(i_l)},a_{\sigma_\ell(i_l+1)}],d\ra \, d
       \] 
       and
       \begin{multline*}
        0\ne\la a_{\sigma_\ell(1)}\cdots [a_{\sigma_\ell(i_l)},a_{\sigma_\ell(i_l+1)}] \cdots a_{\sigma_\ell(m)}, w\ra 
  \\ =   \sum_{d\in\B_L} \la a_{\sigma_\ell(1)}\cdots d \cdots a_{\sigma_\ell(m)}, w\ra\, \la [a_{\sigma_\ell(i_l)},a_{\sigma_\ell(i_l+1)}],d\ra.
       \end{multline*}
       By  the inductive hypothesis $\P(m-1)$, there are only finitely many $(a_{\sigma_\ell(1)},\ldots,d,\ldots,a_{\sigma_\ell(m)})\in(\B_L)^{m-1}$ such that 
       $$\la a_{\sigma_\ell(1)}\cdots d \cdots a_{\sigma_\ell(m)}, w\ra \ne 0,$$
       and for each such choice of $d\in\B_L$ by Assumption \ref{assump: finiteness Lie bracket} there are only finitely many $(a_{\sigma_\ell(i_l)},a_{\sigma_\ell(i_l+1)})\in(\B_L)^2$ such that $$\la [a_{\sigma_\ell(i_l)},a_{\sigma_\ell(i_l+1)}],d\ra\ne 0;$$ from this we obtain the desired finiteness property $\P(m)$.
    \end{enumerate}
    This concludes the proof.
\end{proof}

\begin{assumption}\label{assump: finiteness tr}
For all $c\in \B_L$ the following set is finite: 
\[\Big\{(a,b)\in \B_L\times\B_L: \la a\tr b,c \ra\neq 0\Big\}.\]
\end{assumption}

\begin{lemma}\label{lem: finiteness tr}
Under Assumption \ref{assump: finiteness tr}, for any $c\in\B_L$ the following set is finite:
\[\left\{(a_1,\ldots,a_n,b)\in(\B_L)^{n+1}: \ \la (a_1\cdots a_n)\tr b,c \ra\neq 0\right\}.\]
\end{lemma}
\begin{proof}
    First we prove by induction on $n\geq 1$ the assertion:  
\begin{align*}
    \P(n): ~"&~ \text{for every $c\in \B_L$, the following set is finite:}\\
&\left\{(a_1,\ldots,a_n,b)\in(\B_L)^{n+1}: \ \la (a_1\cdots a_n)\tr b,c \ra\neq 0\right\}"
\end{align*}
 
If $n=1$, then $\P(1)$ is the Assumption \eqref{assump: finiteness tr}.   
Suppose that $\P(n)$ is true for a certain $n\geq 1$. For $(a,a_1,\ldots,a_n,b,c)\in(\B_L)^{n+3}$ we set $v:=a_1\cdots a_n$. 
By Proposition \ref{prop: extension post-lie product} and by linearity:
$$\la (av)\tr b,c\ra = \la a\tr(v\tr b),c\ra - \la(a\tr v)\tr b,c\ra.$$
Therefore:
$$\la (av)\tr b,c\ra\neq 0~\Rightarrow~ \la a\tr(v\tr b),c\ra\neq 0\quad \vee \quad \la(a\tr v)\tr b,c\ra\neq 0.$$

For the first term, we know from \eqref{eq: grafting word on word} that $v\tr b\in L$, and by definition of the pairing $\la\cdot,\cdot\ra$, one has that: 
$$v\tr b=\sum_{d\in\B_L}\la v\tr b,d \ra d.$$
Then one can write by linearity:
$$\la a\tr(v\tr b),c\ra = \sum_{d\in\B_L} \la v\tr b,d \ra \la a\tr d,c\ra.$$
By $\P(1)$, there exists finitely many couples $(a,d)\in \B_L^2$ such that $\la a\tr d,c\ra\neq 0$ and for every such couple $(a,d)$, by $\P(n)$, there exists finitely many $(a_1,\ldots,a_n,b)\in(\B_L)^{n+1}$ such that $\la (a_1\cdots a_n)\tr b,d \ra\neq 0$. We deduce that there exist finitely many $(a,a_1,\ldots,a_n,b,d)\in(\B_L)^{n+3}$ such that:
$$ \la (a_1\cdots a_n)\tr b,d \ra \la a\tr d,c\ra\neq 0,$$
and therefore finitely many $(a_1,\ldots,a_n,b)\in(\B_L)^{n+2}$ such that: $$\la a\tr((a_1\cdots a_n)\tr b),c\ra\neq 0.$$
For the second term, we have $a\tr a_i=\sum_{v\in\B_L} \la a\tr a_i,d\ra \, d$  and then 
by \eqref{eq:derivation} 
\[
\begin{split}
a\tr v = a\tr(a_1\cdots a_n)&= \sum_{i=1}^n a_1 \cdots (a\tr a_i) \cdots a_n  
\\ &= \sum_{i=1}^n \sum_{d\in\B_L} \la a\tr a_i,d\ra   (a_1 \cdots d \cdots a_n),
\end{split}
\]
where in the last sum, $d\in \B_L$ has the $i$-th position, so that
\[
\la (a\tr v)\tr b,c\ra = \sum_{i=1}^n \sum_{d\in\B_L} \la a\tr a_i,d\ra  \la (a_1 \cdots d \cdots a_n)\tr b,c\ra.
\]
By $\P(n)$, for every $i=1,\ldots,n$, the set of all $(a_1,\ldots,d,\ldots,a_n,b)\in(\B_L)^{n+1}$, such that $\la(a_1 \cdots d \cdots a_n)\tr b,c\ra\ne 0$ is finite, and by $\P(1)$, for all such $d\in \B_L$, there exist finitely many couples $(a,a_i)\in \B_L^2$ such that $\la a\tr a_i,d\ra\neq 0$ is finite; therefore $\P(n+1)$ follows.
\end{proof}

\begin{corollary}\label{lem:wd}
Under Assumptions \ref{assump: finiteness Lie bracket} and \ref{assump: finiteness tr}, for any $w\in\Bbar_{\envU(L)}$ the following set is finite:
\[\left\{(u,v)\in \B_{\envU(L)}\times\B_{\envU(L)}:\,\la u\trbar v,w \ra\neq 0\right\}\]
\end{corollary}

\begin{proof} 
Let $w\in\Bbar_{\envU(L)}$. If $w=\ind$, then the only possibility is $u=v=\ind$.
Now, let suppose that $\ell(w)\geq 1$ and consider
$a_1,\ldots,a_m,b_1,\ldots, b_n\in \B_L$. Using first \eqref{eq: post-Lie associative product} and then \eqref{eq: grafting word on word}, we have
\begin{multline*}
    a_1\cdots a_m \star b_1\cdots b_n = \sum_{I_0\sqcup\cdots\sqcup I_n=\{1,\ldots,m\}}a_{I_0}(a_{I_1}\tr b_1) \cdots (a_{I_n}\tr b_n)\\
    = \sum_{\substack{I_0\sqcup\cdots\sqcup I_n=\{1,\ldots,m\}\\d_1,\ldots,d_n\in \B_L}}a_{I_0}\la a_{I_1}\tr b_1, d_1\ra\cdots \la a_{I_n}\tr b_n, d_n\ra d_1\cdots d_n.
\end{multline*}
Then
\begin{multline*}
    \la a_1\cdots a_m \trbar b_1\cdots b_n , w\ra\\
    =\sum_{\substack{I_0\sqcup\cdots\sqcup I_n=\{1,\ldots,m\}\\d_1,\ldots,d_n\in \B_L}} \la a_{I_1}\tr b_1, d_1\ra\cdots \la a_{I_n}\tr b_n, d_n\ra \la a_{I_0} d_1\cdots d_n , w\ra.
\end{multline*}
Then, the last sum is non-zero if at least one of its terms is non-zero, and it is easy to conclude using Lemmas \ref{lem: finiteness Lie bracket} and \ref{lem: finiteness tr}.
\end{proof}

Under Assumptions \ref{assump: finiteness Lie bracket} and \ref{assump: finiteness tr}, the coproduct $\Delta_{\trbar}:\envU(L)\to\envU(L)\otimes\envU(L)$ given by
\begin{equation}\label{eq: Delta_trbar}
 \Delta_{\trbar}v=\sum_{u_1,u_2\in \B_{\envU(L)}} \la u_1\trbar  u_2,v\ra (Tu_1)\otimes (Tu_2), \qquad v\in\envU(L),
\end{equation}
is well-defined thanks to Corollary \ref{lem:wd}.

\begin{prop}\label{prop: dual Hopf algebra}
If Assumptions \ref{assump: finiteness Lie bracket} and \ref{assump: finiteness tr} are satisfied, then the Hopf algebra $(\envU(L),\trbar,\Delta_\ast,\ind,\varepsilon)$ can be
dualized into the Hopf algebra $(\envU(L),\ast,\Delta_{\trbar},{\mathds 1},\varepsilon)$ via the pairing \eqref{eq:dualityL}.
\end{prop}
\begin{proof}
Every connected filtered bialgebra is a filtered Hopf algebra, see \cite[Corollary 5]{MHandbook} or
\cite[Theorem 3.4]{GG19}; since $(\envU(L),\trbar,\Delta_\ast,\ind,\varepsilon)$ is endowed with the filtration \eqref{eq:filtration},
it is a Hopf algebra.

The  bialgebra structure of $(\envU(L),\ast,\Delta_{\trbar},{\mathds 1},\varepsilon)$ is given by duality with the bialgebra $(\envU(L),\trbar,\Delta_\ast,\ind,\varepsilon)$ by reversing the arrows in the defining commutative diagrams. The existence of an antipode for $(\envU(L),\ast,\Delta_{\trbar},{\mathds 1},\varepsilon)$ follows also by duality: if $S$ is the antipode of  $(\envU(L),\trbar,\Delta_\ast,\ind,\varepsilon)$, then $S^*:\envU(L)\to\envU(L)$ defined by
 \[
 S^*v:=\sum_{u\in \B_{\envU(L)}} \la Su, v\ra T u,
 \]
 is an antipode for $(\envU(L),\ast,\Delta_{\trbar},{\mathds 1},\varepsilon)$.
 \end{proof}

\begin{remark}
A particular case for which Assumptions \ref{assump: finiteness Lie bracket} and \ref{assump: finiteness tr} are trivially satisfied is when $L$ is "graded of finite type", that is to say when $L=\bigoplus_{n=0}^\infty L_n$ with $\dim(L_n)<\infty$, and the operations $\tr$ and $[\cdot,\cdot]$ respect the gradation, that is to say their restrictions are mapping $L_p\otimes L_q$ into $L_{p+q}$.
Two particular instances of such graded post-Lie algebras of finite type in the literature are:
\begin{itemize}[leftmargin=*]
    \item The free pre-Lie algebra, being the free vector space on the set of (decorated) rooted trees, whose grading is given by the number of vertices, endowed with the grafting product, see for example \cite{CL}, which is the framework for Branched Rough Paths theory \cite{ramification}.
    \item The free post-Lie algebra, being the free Lie algebra on the set of planary (decorated) rooted trees, whose grading is given by the number of vertices, endowed with the grafting product, extended on formal Lie brackets using the axioms of post-Lie algebras, see \cite{munthe2013post}, which is the framework for planarly branched rough paths \cite{CEFMM}.
\end{itemize}
However, we emphasize that this hypothesis of finite type will not be satisfied in our context, which motivates the need for Assumptions \ref{assump: finiteness Lie bracket} and \ref{assump: finiteness tr}.
\end{remark}

\subsection{The character group} 
We note that for all $v\in\envU(L)$ we have
\begin{equation*}
v=\sum_{u\in\B_{\envU(L)}}  \la u,v\ra \, Tu,
\end{equation*}
where $T:\envU(L)\to\envU(L)$ is the linear operator defined in \eqref{eq:T}.

We define the (real) dual space $\envU(L)^*$ as the space of linear maps $f:\envU(L)\to \R$. As before, we consider a basis $\B_L$ of $L$ and a total order $\leq$ on it and the PBW Theorem \ref{theo: PBW} induces the basis $\B_{\envU(L)}$ in \eqref{eq:basis} and
${\overline\B}_{\envU(L)}$ in \eqref{eq:basis*} of $\envU(L)$.
Then for all $f\in\envU(L)^*$ and all $v\in\envU(L)$, we have
$$f(v)=\sum_{u\in\B_{\envU(L)}}  \la u,v\ra \, f(Tu).$$
This allows us to identify $\envU(L)^*$ with a space of formal series
\begin{equation}\label{eq: duality vs formal series}
\begin{gathered}
 \envU(L)^*\ni f\longleftrightarrow \sum_{u\in\B_{\envU(L)}} f(Tu)\, u\in\left\{\sum_{u\in\B_{\envU(L)}}\alpha_u u:\,\alpha_v\in \R\right\}, \\
    \left(\sum_{u\in\B_{\envU(L)}}\alpha_u u\right)(v):=\sum_{u\in\B_{\envU(L)}} \alpha_u  \la u,v\ra.
        \end{gathered}
\end{equation}

\begin{definition}\label{def:grch}
The set $G\subset\envU(L)^*$ of (real-valued) \textbf{characters} on $(\envU(L),\ast)$ is defined as the set of $\ast$-multiplicative linear forms $f\in\envU(L)^*$ such that $f(\mathds{1})=1$
\[
f(u_1\ast u_2) = f(u_1)\, f(u_2), \qquad u_1,u_2\in\envU(L).
\]
We also define $H:=\{f\in\envU(L)^*: f(\mathds{1})=1\}$.
\end{definition}

If Assumptions \ref{assump: finiteness Lie bracket} and \ref{assump: finiteness tr} are satisfied, we have proved in Proposition \ref{prop: dual Hopf algebra} that $(\envU(L),\ast,\Delta_{\trbar},\ind,\varepsilon)$ is a Hopf algebra. This leads to the following well-known result (see for example \cite[Proposition 19]{MHandbook}):
\begin{prop}\label{pr:grch}
    If Assumptions \ref{assump: finiteness Lie bracket} and \ref{assump: finiteness tr} are satisfied, the set $H$ in Definition \ref{def:grch}
    can be endowed with a group structure $(\trbar,\ind^*)$, where the unit element is given by duality as $\mathds{1}^*(\cdot):=\la \mathds{1},\cdot\ra$, the product is given by:
\begin{equation*}
f_1\trbar  f_2  := m_\R(f_1\otimes f_2)\Delta_\trbar.
\end{equation*}
where $m_\R$ denotes the multiplication in $\R$, and the inverse of $f\in H$ is computed as:
$$f^{-1}=\sum_{n\geq 0} (\varepsilon-f)^{\star n}.$$

Moreover the set $G$ of characters is a subgroup of $H$.
\end{prop}
Using the identification \eqref{eq: duality vs formal series}, we can also write
\begin{equation*}
(f_1\trbar  f_2)(v)=\sum_{u_1,u_2\in\B_{\envU(L)}} f_1(Tu_1)f_2(Tu_2)\la u_1\trbar u_2,v\ra.
\end{equation*}
\bigskip

\section{The post-Lie algebra of derivations}\label{sec:derivations}
\medskip
\subsection{Derivations and post-Lie algebra structure} In this section, we use the notations of \cite{ebrahimi2014lie}. We fix once and for
all an associative and commutative $\K$-algebra $(\A,\cdot)$.

The space of derivations $\Der(\A)$ on $\A$ is the subspace of all $D\in\End(\A)$ satisfying the following Leibniz rule that for all $a,b\in\A$
\[D(a\cdot b)=D(a)\cdot b + a\cdot D(b).\]
One can easily prove a generalisation of the Leibniz rule for products of $n\geq 2$ terms,
by induction on $n$, given for all $a_1,\ldots,a_n\in\A$ by
\begin{equation*}
    D(a_1\cdots a_n)=\sum_{i=1}^n a_1\cdots D(a_i)\cdots a_n.
\end{equation*}

One of the most common examples for the algebra $\A$ is the space of smooth functions $C^\infty(\K^n)$, where $\K=\R~\text{or}~\C$ endowed with the pointwise product. In particular polynomials in $\K[z_1,\ldots,z_n,\ldots]$ fulfill that condition and each derivation $D$ on that algebra is given as formal series of partial derivations along each coordinate $\partial_{z_i}$:
$$D=\sum_{i}D(z_i)\partial_{z_i}.$$

Another relevant example in our setting is the following: given a post-Lie algebra $(L,\tr,[\cdot,\cdot])$ and the universal enveloping algebra $\env_{[\cdot,\cdot]}(L)$, then every element $a\in L$ defines a derivation on $\env_{[\cdot,\cdot]}(L)$ via the extension of $\tr$ to $\env_{[\cdot,\cdot]}(L)$, see \eqref{eq:derivation}. \\

In the following, we will denote by $\circ$, as usual, the composition operation in
$\End(\A)$. The \textit{commutator} of $\circ$ is the anti-commutative binary operation on $\End(\A)$ defined by
\begin{equation*}
    [D_1,D_2]_\circ=D_1\circ D_2 - D_2\circ D_1.
\end{equation*}

This is a Lie bracket by associativity of $\circ$, which moreover stabilises $\Der(\A)$, since
for all $a,b\in\A$, we have
\begin{align*}
    D_1\circ D_2 (ab) &= D_1\circ \Big(D_2 (a)b + aD_2 (b)\Big)\\
    &= D_1\circ D_2 (a)b + D_2 (a)D_1 (b) + D_1 (a)D_2 (b) + aD_1\circ D_2 (b).
\end{align*}
Thus, after inverting the indices, one obtains
$$[D_1,D_2]_\circ (ab) = [D_1,D_2]_\circ (a)b + a[D_1,D_2]_\circ (b),$$
which proves that $[D_1,D_2]_\circ\in \Der(\A)$.\\

For $a\in\A$ and $D\in\End(\A)$ we denote 
\begin{equation}\label{eq:aD}
a\cdot D:\A\to\A, \qquad a\cdot D(b):=aD(b). 
\end{equation}
If $D\in\Der(\A)$ then $a\cdot D$ also belongs to $\Der(\A)$.\\

The main tool of the article is the following:
\begin{theorem}\label{theo: post-Lie structure from derivations}
Let $\D$ be a sub-Lie algebra of $\Der(\A)$ for the commutator bracket $[\cdot,\cdot]_\circ$. The vector space $\A\otimes \D$ admits a structure of (left) post-Lie algebra $(\tr,[\cdot,\cdot])$ given for all $a_1,a_2\in\A$ and $D_1,D_2\in \D$ by:
\begin{equation}\label{eq: post-lie product derivation}
    a_1\otimes D_1\tr a_2 \otimes D_2:=a_1 D_1(a_2)\otimes D_2,
\end{equation}
\begin{equation}\label{eq: post-lie bracket derivation}
    [a_1\otimes D_1, a_2 \otimes D_2]:=a_1a_2\otimes [D_1,D_2]_\circ.
\end{equation}
\end{theorem}
\begin{proof} Let us first compute the value of the associator of $\tr$.
Take $a_1,a_2,a_3\in\A$ and $D_1,D_2,D_3\in \D$.
On one hand by the Leibniz rule:
\begin{align*}
    a_1\otimes D_1&\tr\left(a_2\otimes D_2\tr a_3\otimes D_3\right)\\
    &=a_1\otimes D_1\tr\left(a_2 D_2(a_3)\otimes D_3\right)\\
    &=a_1 D_1(a_2)  D_2(a_3) \otimes D_3+ a_1 a_2 D_1\circ D_2(a_3)\otimes D_3.
\end{align*}
On the other hand:
\begin{align*}
    \left(a_1\otimes D_1\tr a_2\otimes D_2\right)\tr a_3\otimes D_3&=\left(a_1 D_1(a_2)\otimes D_2\right)\tr a_3\otimes D_3\\
    &=a_1D_1(a_2) D_2(a_3)\otimes D_3.
\end{align*}
By subtracting the last two equalities one finally obtains that
$$\ass_\tr(a_1\otimes D_1,a_2\otimes D_2,a_3\otimes D_3)=a_1 a_2 D_1\circ D_2(a_3)\otimes D_3.$$
Now let us verify the two post-Lie conditions. By commutativity of $\A$, one has:
\begin{align*}
    &\ass_\tr(a_1\otimes D_1,a_2\otimes D_2,a_3\otimes D_3)-\ass_\tr(a_2\otimes D_2,a_1\otimes D_1,a_3\otimes D_3)\\
    &=a_1 a_2 D_1\circ D_2(a_3)\otimes D_3 - a_1 a_2 D_2\circ D_1(a_3)\otimes D_3\\
   & =(a_1a_2\otimes [D_1,D_2]_\circ) \tr (a_3\otimes D_3)
    \\ &= [a_1\otimes D_1, a_2 \otimes D_2] \tr a_3\otimes D_3.
\end{align*}
by the definition \eqref{eq: post-lie bracket derivation} of $[\cdot,\cdot]$.
Finally, by the definitions \eqref{eq: post-lie product derivation} and \eqref{eq: post-lie bracket derivation} of $[\cdot,\cdot]$ and $\tr$, one has:
\[
\begin{split}
  &  a_1\otimes D_1 \tr [a_2\otimes D_2 , a_3\otimes D_3]
  \\ & =a_1 D_1(a_2)a_3\otimes [D_2,D_3]_\circ + a_1a_2D_1(a_3)\otimes [D_2,D_3]_\circ\\
  &  =[a_1D_1(a_2)\otimes D_2,a_3\otimes D_3]+[a_2\otimes D_2,a_1D_1(a_3)\otimes D_3]\\
  &  =[a_1\otimes D_1\tr a_2\otimes D_2,a_3\otimes D_3]+[a_2\otimes D_2,a_1\otimes D_1\tr a_3\otimes D_3].
\end{split}
\]
The proof is complete.
\end{proof}

\begin{corollary}[Burde]\label{coro: pre-Lie structure of commutative derivations}
    If $\mathcal{D}\subset \Der(\A)$ is a linear space of derivations which commute with each other for the composition product, then $(\A\otimes \D,\tr)$ is a left pre-Lie algebra, where the pre-Lie product $\tr$ is given by the formula \eqref{eq: post-lie product derivation}.
\end{corollary}
The latter result is Proposition 2.1 in \cite{Burde}, where left pre-Lie algebras are called {\it left-symmetric algebras}.

\medskip
We give now an extension of Definition \ref{def:preLie}, namely the notion of multiple pre-Lie algebras, see \cite{foissy18}.

\begin{definition}
A (left) \textbf{multiple pre-Lie algebra} $(\A,\{\tr_i\}_{i\in I})$ is the data of a vector space $\A$, endowed with a family of bilinear operations $\tr_i:\A\otimes \A\rightarrow \A$ indexed by a set $I$, which verifies the following relation for all $i,j\in I$ and $a,b,c\in \A$:
\begin{equation*}
    a\tr_i (b\tr_j c) - (a\tr_i b)\tr_j c = b\tr_j (a\tr_i c) - (b \tr_j a)\tr_i c.
\end{equation*}
If the index set $I$ is a singleton, namely $\{\tr_i\}_{i\in I}=\{\tr\}$, then the data $(\A,\tr)$ is a (left) \textbf{pre-Lie algebra},
namely a particular case of Definition \ref{def:preLie}. 
\end{definition}

Then we have the following
\begin{corollary} Let $\{D_i\}_{i\in I}\subset \Der(\A)$ a set of commuting derivations. The family of binary operations $\{\tr_i\}_{i\in I}$ defined for all $a,b\in \A$ and $i\in I$ by:
    $$a\tr_i b:=aD_i(b)$$
  makes $(\A,\{\tr_i\}_{i\in I})$ a multiple left pre-Lie algebra.
\end{corollary}
\begin{proof}
It is a direct application of Corollary \ref{coro: pre-Lie structure of commutative derivations}, where $\mathcal{D}$ is the linear space of derivations generated by $\{D_i\}_{i\in I}$.
\end{proof}

\begin{corollary}\label{coro: pre-Lie product implied by single derivation}
    Every derivation $D\in \Der(\A)$ defines a pre-Lie product $\tr$ on $\A$ given for all $a,b\in \A$ by:
    $$a\tr b=a D(b).$$
\end{corollary}

\begin{remark} 
    Let $\A$ be a space endowed with a set $\{\tr_i\}_{i\in I}$ of binary operations $\A\otimes \A \to \A$ indexed by a set $I$ and denote $\R.I$ the free real vector space generated by it. Consider the tensor product of vector spaces $\A\otimes \R.I$, endowed with the binary operation $\tr$ defined by:
    $$(a\otimes i)\tr (b\otimes j)=a\tr_i b \otimes j$$
    Then it is an easy exercise to show that $(\A,\{\tr_i\}_{i\in I})$ is a multiple pre-Lie algebra if and only if $(\A\otimes \R.I,\tr)$ is a pre-Lie algebra.
\end{remark}
\medskip

\subsection{Associative product on the post-Lie enveloping algebra}
Note that Theorem \ref{theo: post-Lie structure from derivations} applies in particular to $\A\otimes\Der(\A)$, which is therefore
endowed with a natural post-Lie structure. We fix a sub-post-Lie algebra $L\subseteq \A\otimes\Der(\A)$.

Following Proposition \ref{prop: extension post-lie product} we know that an extension of the post-Lie product $\tr$ to $\env_{[\cdot,\cdot]}(L)$ can be constructed. We make explicit the extension on the left:

\begin{prop}\label{prop: grafting a word on a letter derivations}
The extension of the post-Lie product $\tr$ to $\env_{[\cdot,\cdot]}(L)$ as in Proposition \ref{prop: extension post-lie product}, is given on the left by:
\begin{equation}\label{eq:extensionlefttr}
    \Big((a_1\otimes D_1)\cdots (a_n\otimes D_n)\Big)\tr (a \otimes D)= \Big(a_1\cdots a_n\cdot D_1\circ \ldots \circ D_n (a)\Big) \otimes D.
\end{equation}
\end{prop}
\begin{proof} The equality is trivially verified if $n=1$ by the definition \eqref{eq: post-lie product derivation} of $\tr$ on $\A\otimes\D$. Suppose that it is verified for all 
words of length up to a fixed integer $n-1$. By equality $(2)$ of Proposition \ref{prop: extension post-lie product}, setting $u:=(a_2\otimes D_2)\cdots (a_n\otimes D_n)$, we have
\begin{multline*}
 \Big((a_1\otimes D_1)\cdots (a_n\otimes D_n)\Big)\tr (a \otimes D)\\
 = (a_1\otimes D_1)\tr \Big(u \tr (a \otimes D)\Big)- \Big((a_1\otimes D_1) \tr u\Big) \tr (a \otimes D).
\end{multline*}

For the first term in the latter expression, using the inductive hypothesis, one has:
\begin{align*}
    (a_1\otimes D_1) \tr \Big(u\tr (a \otimes D)\Big)=&(a_1\otimes D_1) \tr \Big(\big(a_2\cdots a_n D_2\circ \cdots \circ D_n (a)\big) \otimes D\Big) \\
    =&\left(\sum_{i=2}^n a_1 a_2 \cdots D_1(a_i)\cdots a_n D_2 \cdots  D_n (a)\right)\otimes D\\
    &+\Big(a_1\cdots a_n D_1 \ldots  D_n (a)\Big)\otimes D
\end{align*}
For the second term, using the \eqref{eq:derivation} and the inductive hypothesis, one has:
\begin{align*}
    \Big((a_1&\otimes D_1) \tr u\Big) \tr (a \otimes D)\\
    &=\left(\sum_{i=2}^n(a_2\otimes D_2)\cdots (a_1 D_1(a_i)\otimes D_i)\cdots (a_n\otimes D_n)\right)\tr (a \otimes D) \\
    &=\left(\sum_{i=2}^n a_1 a_2 \cdots D_1(a_i)\cdots a_n D_2 \circ \cdots \circ D_n (a)\right)\otimes D.
\end{align*}
The proof is concluded by subtracting the two previous expressions.
\end{proof}

We shall use in the following the analog of the notation \eqref{eq:notationaI} for $I\subset\{1,\ldots,n\}$ and $D_1,\ldots,D_n\in\Der(\A)$
\begin{equation*}
D_I:=D_{i_1} \circ \cdots \circ D_{i_p}, \qquad I=\{i_1,\ldots,i_p\}, \qquad i_1<\cdots<i_p,
\end{equation*}
and $D_\emptyset:={\rm Id}_\A$.

By Proposition \ref{prop: associative product} we can endow $\env_{[\cdot,\cdot]}(L)$ with an associative product $\trbar$ defined by \eqref{eq: post-Lie associative product}. In particular for all $a_1\otimes D_1,a_2\otimes D_2\in L$, we have
$$(a_1\otimes D_1) \star (a_2\otimes D_2)=(a_1\otimes D_1)(a_2\otimes D_2) + a_1 D_1(a_2)\otimes D_2.$$

More generally we have the following proposition:
\begin{prop}\label{prop: extension product}
The relation \eqref{eq: grafting word on word} completes the extension \eqref{eq:extensionlefttr} of $\tr$ on the right, yielding the explicit expression for the extension of the post-Lie product $\tr$ on $\env_{[\cdot,\cdot]}(L)$:
\begin{gather*}
\Big((a_1\otimes D_1)\cdots (a_n\otimes D_n)\Big) \tr \Big((\tilde a_1\otimes \tilde D_1)\cdots (\tilde a_m\otimes \tilde D_m)\Big) 
\\ = \sum_{I_1\sqcup\cdots \sqcup I_m=\{1,\ldots,n\}}  \prod_{j=1}^m \left( a_{I_j} D_{I_j} (\tilde a_j) \otimes \tilde D_j\right).
\end{gather*}

Analogously we obtain the explicit expression for the associative product $\trbar$ on $\env_{[\cdot,\cdot]}(L)$ as in \eqref{eq: post-Lie associative product}:
\begin{align}
 \nonumber &\Big((a_1\otimes D_1)\cdots (a_n\otimes D_n)\Big) \star \Big((\tilde a_1\otimes \tilde D_1)\cdots (\tilde a_m\otimes \tilde D_m)\Big)
\\ \nonumber & = \sum_{I\sqcup J=\{1,\ldots,n\}} \prod_{i\in I} (a_i \otimes D_i) \left[\left( \prod_{j\in J} (a_j \otimes D_j)\right)\tr \Big((\tilde a_1\otimes \tilde D_1)\cdots (\tilde a_m\otimes \tilde D_m)\Big)\right]
\\ & = \sum_{I\sqcup J_1\sqcup\cdots\sqcup J_m=\{1,\ldots,n\}} \prod_{i\in I} (a_i \otimes D_i) \prod_{j=1}^m \left( a_{J_j} D_{J_j} (\tilde a_j) \otimes \tilde D_j\right).
\label{eq:trbargeneral}
\end{align}
\end{prop}


\subsection{Representation of the enveloping algebras}\label{subsection: Representation of the enveloping algebras}

We still consider a sub-post-Lie algebra $L\subseteq \A\otimes\Der(\A)$, endowed with
the post-Lie structure given in Theorem \ref{theo: post-Lie structure from derivations}.
In this section we aim at giving algebra representations of $(\env_{[\cdot,\cdot]}(L),\trbar)$ and $(\env_{\llbracket \cdot,\cdot \rrbracket}(L),\mathsf{conc})$ on $\A$, that is to say algebra morphisms with values in the space of endomorphisms $\End(\A)$ endowed with the composition product $\circ$.

Consider the linear map $\rho:\A\otimes \Der(\A)\to \Der(\A)$ given by 
\begin{equation}\label{eq:rho}
\rho(a\otimes D)=a\cdot D,
\end{equation} 
where $a\cdot D$ denotes the element of $\End(\A)$ defined in \eqref{eq:aD}.
We have seen before, in Theorem \ref{theo: post-Lie structure from derivations} that $(\Der(\A),[\cdot,\cdot]_\circ)$ is a sub-Lie algebra of $(\End(\A),[\cdot,\cdot]_\circ)$, while by Proposition \ref{pr:ebrahimi2014lie} $(L,\llbracket\cdot,\cdot\rrbracket)$ is a
Lie algebra since $L\subseteq \A\otimes\Der(\A)$ is post-Lie. The relation between these two Lie algebras is explained by the following:
\begin{lemma}
The map $\rho:(L,\llbracket\cdot,\cdot\rrbracket)\to (\Der(\A),[\cdot,\cdot]_\circ)$ is a morphism of Lie algebras.
\end{lemma}
\begin{proof}
The composition Lie bracket defined by equality \eqref{eq: composition Lie bracket} is equal on $\A\otimes\Der(\A)$ to:
\begin{equation*}
    \llbracket a_1\otimes D_1,a_2\otimes D_2\rrbracket=a_1D_1(a_2)\otimes D_2 - a_2D_2(a_1)\otimes D_1+a_1a_2\otimes[D_1,D_2]_\circ.
\end{equation*}
On the other hand for all $a_1,a_2\in \A$ and $D_1,D_2\in \Der(\A)$, we have
\begin{equation*}
    [a_1\cdot D_1,a_2\cdot D_2]_\circ=a_1D_1(a_2)D_2 - a_2D_2(a_1)D_1+a_1a_2[D_1,D_2]_\circ.
\end{equation*}
The proof is complete.
\end{proof}

\begin{remark}
Given a sub-Lie algebra $\D$ of $(Der(\A),[\cdot,\cdot]_\circ)$, we can endow $\A\otimes\D$ with a structure of $\A$-module with the action of $\A$ on $A\otimes\D$ being given for all $a,b\in\A$ and $ D\in \D$ by:
$$a\bullet(b\otimes D) := (a b)\otimes D$$ 
It is easy to show the following Leibniz rule for all $u,v\in \A\otimes\D$ and $a\in\A$:
$$\llbracket u, a\bullet v \rrbracket = (\rho(u)[a])\bullet v + a\bullet \llbracket u, v\rrbracket.$$
This turns $(\A\otimes\D,\tr,[\cdot,\cdot],\rho)$ into a $(\A,\bullet)$-post-Lie--Rinehard algebra (it seems that this is
actually the first example of a post-Lie--Rinehard algebra which is not pre-Lie). 
For more details on pre-Lie algebras in the context of aromatic B-series, the reader can refer to \cite{floystad2021universal}.
\end{remark}

Note that $\rho$ is a representation of $(L,\llbracket\cdot,\cdot\rrbracket)$ on $\A$. 
By the universal property of $\env_{\llbracket \cdot,\cdot \rrbracket}(L)$, 
it can be extended uniquely to a morphism $\rhohat$ of associative algebras
\begin{equation}\label{eq:rhohat}
\writefun{\rhohat}{ \left(\env_{\llbracket \cdot,\cdot \rrbracket}(L),\mathsf{conc}\right)}{\left(\End(\A),\circ\right)}{(a_1\otimes D_1)\cdots(a_n\otimes D_n)}{(a_1\cdot D_1)\circ\cdots\circ(a_n\cdot D_n)}
\end{equation}
Then $\rhohat$ is a representation of $\left(\env_{\llbracket \cdot,\cdot \rrbracket}(L),\mathsf{conc}\right)$ on $\A$ which extends $\rho:(L,\llbracket\cdot,\cdot\rrbracket)\to (\Der(\A),[\cdot,\cdot]_\circ)$
given by \eqref{eq:rho}.

However we are interested rather in an extension of $\rho$ to a morphism of algebras 
$\rhobar:(\env_{[\cdot,\cdot]}(L),\trbar)\to(\End(\A),\circ)$.
Theorem \ref{theo: isomorphism between enveloping algebras} states that the linear map $\Phi:\env_{\llbracket \cdot,\cdot \rrbracket}(L)\rightarrow \env_{[\cdot,\cdot]}(L)$ defined for $a_1,\ldots,a_n \in L$ by:
\begin{equation*}
    \Phi(a_1 \cdots a_n):=a_1\trbar \ldots \trbar  a_n
\end{equation*}
is an algebra isomorphism between $\left(\env_{\llbracket \cdot,\cdot \rrbracket}(L),\mathsf{conc}\right)$ and $\left(\env_{[\cdot,\cdot]}(L),\trbar\right)$.
This isomorphism allows to give an extension of the representation $\rho$  to a representation:
\begin{equation}\label{eq: representation morphism}
\rhobar:\left(\env_{[\cdot,\cdot]}(L),\trbar\right)\to(\End(\A),\circ), \qquad \rhobar=\rhohat\circ\Phi^{-1},
\end{equation}
namely we have the following commutative diagram of associative algebras in which the dashed arrows represent morphisms of Lie algebras and plain arrows represent morphisms of associative algebras:
\[
\begin{tikzcd}[scale cd=1,sep=large]
    & (L,\llbracket\cdot,\cdot\rrbracket) \arrow[dl,hook',dashed]\arrow[dd,dashed] \arrow[dr,hook,dashed] \\
    \left(\env_{\llbracket \cdot,\cdot \rrbracket}(L),\mathsf{conc}\right) \arrow[rr,"\Phi~\sim", pos=0.4, crossing over] \arrow[swap]{dr}{\rhohat} && \left(\env_{[\cdot,\cdot]}(L),\trbar\right) \arrow{dl}{\rhobar} \\
     & (\End(\A),\circ)
  \end{tikzcd}
\]

This representation can be made more explicit:
\begin{theorem}\label{theo: canonical representation} The linear map $\rhobar$ defined in \eqref{eq: representation morphism}
admits the following explicit expression
\begin{equation}\label{eq: representation algebra derivation}
    \rhobar\Big((a_1\otimes D_1)\cdots (a_n\otimes D_n)\Big)= a_1\cdots a_n\cdot (D_1\circ \ldots \circ D_n).
\end{equation}
By the algebra morphism property for $\rhobar:\left(\env_{[\cdot,\cdot]}(L),\trbar\right)\to(\End(\A),\circ)$, we also have
\begin{equation}\label{eq: representation algebra derivation2}
\bar\rho\Big((a_1\otimes D_1)\trbar \cdots\, \trbar\,(a_n\otimes D_n)\Big)=(a_1\cdot D_1)\circ\cdots\circ(a_n\cdot D_n).
\end{equation}

\end{theorem}

\begin{proof}
Setting $\rhobar=\rhohat\circ\Phi^{-1}$ as in \eqref{eq: representation morphism}, we obtain automatically that $\rhobar$ is
a morphism of algebras and therefore that it is the unique extension of $\rho:L\rightarrow \Der(\A)$ to
a morphism of algebras $\left(\env_{[\cdot,\cdot]}(L),\trbar\right)\to(\End(\A),\circ)$. 

We want now to show that $\rhobar$ satisfies \eqref{eq: representation algebra derivation}. We proceed by induction on $n$; for
$n=1$ the claim follows from the definition \eqref{eq:rho} of $\rho$. Let us suppose now that  \eqref{eq: representation algebra derivation}
is proved for $n\geq 1$; let us set for ease of notation $u_i:=a_i\otimes D_i$, $i=0,\ldots,n$; then by \eqref{eq:atrbbb}, we have
\[
u_0\trbar (u_1\cdots u_n) = \sum_{i=1}^n u_1 \cdots (u_0\tr u_i) \cdots u_n+ u_0 \cdots u_n.
\]
By the definition \eqref{eq: post-lie product derivation} of $\tr$ we have $u_0\tr u_i=a_0D_0(a_i)\otimes D_i$. 
By applying $\rhobar$ we obtain by the induction hypothesis
\[
\rhobar(u_0\trbar (u_1\cdots u_n)) = \sum_{i=1}^n a_0a_1\cdots D_0(a_i) \cdots a_n \cdot (D_1\circ \ldots \circ D_n)
+ \rhobar(u_0 \cdots u_n).
\]
On the other hand, by the morphism property and the induction hypothesis
\[
\begin{split}
\rhobar(u_0\trbar (u_1\cdots u_n)) =&
\rhobar(u_0)\circ \rhobar(u_1\cdots u_n)
\\ =& (a_0 \cdot D_0)\circ \left(a_1\cdots a_n\cdot (D_1\circ \ldots \circ D_n)\right)
\\ =& \sum_{i=1}^n a_0a_1\cdots D_0(a_i) \cdots a_n \cdot (D_1\circ \ldots \circ D_n)\\
&+ a_0\cdots a_n\cdot(D_0\circ\cdots\circ D_n).
\end{split}
\]
Therefore we obtain as required
\[
\rhobar(u_0 \cdots u_n)=\rhobar\left((a_0\otimes D_0)\cdots (a_n\otimes D_n)\right)= a_0\cdots a_n\cdot(D_0\circ\cdots\circ D_n)
\]
and the proof is complete.
\end{proof}

\begin{remark}
    One should remark that the representations $\rho$, $\rhohat$ and $\rhobar$ are not faithful, since for example for $a,b\in\A,~ a\neq b$ and $D\in \Der(\A)$, we have
    $$\rho\left(a\otimes (b\cdot D)\right)=\rho \left(b\otimes (a\cdot D)\right).$$
\end{remark}

\begin{remark}\label{rem: left extension post-lie product = rhobar}
By \eqref{eq:extensionlefttr} the left extension of $\tr$ on $\env_{[\cdot,\cdot]}(L)$ can be expressed in terms of the representation $\rhobar$:
    \begin{equation*}
    u\tr (a \otimes D)=\rhobar(u)(a)\otimes D, \qquad u\in\env_{[\cdot,\cdot]}(L).
\end{equation*}
Moreover by \eqref{eq:trbargeneral} for $u\in\env_{[\cdot,\cdot]}(L)$, we have
\begin{equation}\label{eq:trbargeneral2}
u\trbar  (a \otimes D) =\sum_{(u)} u^{(1)} \left[ \rhobar\left( u^{(2)}\right)(a) \otimes D\right]
\end{equation}
and for all $\tilde u = (\tilde a_1\otimes \tilde D_1)\cdots (\tilde a_m\otimes \tilde D_m) \in\env_{[\cdot,\cdot]}(L)$, we have
\begin{equation*}
u\trbar  \tilde u  =\sum_{(u)} u^{(1)} \prod_{i=1}^m \left[ \rhobar\left( u^{(i+1)}\right)(\tilde a_i) \otimes \tilde D_i\right],
\end{equation*}
with the extension of Sweedler's notation \eqref{eq:Sweedler}
\begin{equation*}
\Cop^{(m)} u = \sum_{(u)} u^{(1)}\otimes\ldots\otimes u^{(m+1)},
\end{equation*}
\[
\text{where}\qquad 
\Cop^{(1)}:=\Cop, \qquad \Cop^{(m+1)}:=({\rm id}\otimes\Cop)\Cop^{(m)}.
\]
\end{remark}

\begin{prop}\label{prop: multiplicativity property representation}
    Given $b_1,b_2\in \A$ and $u=(a_1\otimes D_1)\cdots (a_n\otimes D_n)\in \env_{[\cdot,\cdot]}(\A\otimes\Der(\A))$, the Leibniz rule of $D_1,\ldots,D_n$ on $\A$ implies     \begin{align*}
        \rhobar(u)(b_1b_2)&=\sum_{I\sqcup J=\{1,\ldots,n\}} a_I D_I(b_1) a_J D_J(b_2)\\
        &= \sum_{I\sqcup J=\{1,\ldots,n\}}\rhobar\left(\prod_{i\in I}(a_i\otimes D_i)\right)(b_1)\, \rhobar\left(\prod_{j\in J}(a_j\otimes D_j)\right)(b_2)\\
        &=\sum_{(u)} \rhobar(u^{(1)})(b_1) \,  \rhobar(u^{(2)})(b_2)\\
        &= (\rhobar\otimes \rhobar)(\Cop u)(b_1\otimes b_2).
    \end{align*}
\end{prop}

\subsection{The dual coalgebra structure}

For the sake of generality, let us once again consider a commutative and associative algebra $(\mathcal{A},\cdot)$ equipped with a basis 
$\B_\A$, which allows to define a pairing given by the bilinear extension of
\begin{equation}\label{eq:dualityA}
\B_\A\times\B_\A\ni(a,b)\mapsto \la a,b\ra :=\ind_{(a=b)}.
\end{equation}
Let as before $(L,\tr,[\cdot,\cdot])$ be a sub-post-Lie algebra of $\A\otimes\Der(\A)$ for the canonical post-Lie 
structure defined by Theorem \ref{theo: post-Lie structure from derivations}.

We fix a basis $\B_L$ of $L$ composed of elements of type $a\otimes D$ where $a\in \B_\A$ and $D\in\Der(\A)$. We know from Section \ref{subsec: Env. alg.} that given a total order $\leq$ on $\B_L$, the Poincaré-Birkhoff-Witt Theorem \ref{theo: PBW} gives a vectorial basis $\B_{\envU(L)}$ of $\env_{[\cdot,\cdot]}(L)$ composed by monomials
$$u=(a_1\otimes D_1)\cdots (a_k\otimes D_k)$$
where the factors $a_i\otimes D_i$ belong to $\B_L$, and are organized in increasing order.

In order to prove that $\Delta_\trbar:\envU(L)\to \envU(L)\otimes \envU(L)$ given by \eqref{eq: Delta_trbar} is well defined, we need to make the following crucial assumptions on $(\A,L)$.

\begin{assumption}\label{assump: finiteness Lie bracket AtensDer(A)}
The set \[\left\{(a\otimes D,b\otimes D')\in (\B_L)^2:\,\la [a\otimes D,b\otimes D'],c\otimes D''\ra\neq 0\right\}\] is finite for all $c\otimes D''\in\B_L$, where the pairing $\la\cdot,\cdot\ra$ on $\A$ is defined in \eqref{eq:dualityA}.
\end{assumption}
Assumption \ref{assump: finiteness Lie bracket AtensDer(A)} is simply a rewriting of Assumption \ref{assump: finiteness Lie bracket}. 

\begin{assumption}\label{assump: finiteness rhobar}
The set \[\Big\{(a\otimes D,b)\in \B_L\times\B_\A:\,\la \rho(a\otimes D)(b),c\ra\neq 0\Big\}\] is finite for all $c\in\B_\A$, where the pairing $\la\cdot,\cdot\ra$ on $\A$ is defined in \eqref{eq:dualityA}.
\end{assumption}

\begin{lemma}\label{lem: implication assumptions}
    Assumption \ref{assump: finiteness rhobar} on $((L,\B_L),(\A,\B_\A))$ implies Assumption \ref{assump: finiteness tr} on $(L,\B_L)$.
\end{lemma}
\begin{proof}
We distinguish here the pairing $\la\cdot,\cdot\ra_{\envU(L)}$ defined on $\envU(L)$ by \eqref{eq:dualityL} and $\la\cdot,\cdot\ra_{\A}$ defined on $\A$ by \eqref{eq:dualityA}. The basis $\B_L$ being composed by elements of type $a\otimes D$ where $a\in \B_\A$ and $D\in\Der(\A)$, the implication follows easily from the equalities
\begin{align*}
    \la (a\otimes D)\tr(b\otimes D') , c\otimes D''\ra_{\B_L}&=\la aD(b)\otimes D',c\otimes D''\ra_{\B_L}\\
    &=\la aD(b),c\ra_{\B_\A}\ind_{(D'=D'')}\\
    &=\la \rho(a\otimes D)(b),c\ra_{\B_\A} \ind_{(D'=D'')}
\end{align*}
for all $a\otimes D,b\otimes D' , c\otimes D''\in \B_L$.
\end{proof}

\begin{prop}\label{prop: finiteness rhobar} If $(L,\A)$ satisfies Assumption \ref{assump: finiteness rhobar}, then the set
    \[\left\{(u,b)\in\B_{\envU(L)}\times \B_\A:\,\la\rhobar(u)(b),c\ra\ne 0\right\}\]
    is finite for all $c\in\B_\A$, where the pairing $\la\cdot,\cdot\ra$ on $\A$ is defined in \eqref{eq:dualityA}.
\end{prop}
\begin{proof}
    We can prove this Proposition by following the exact same steps as in the proof of Lemma \ref{lem: finiteness tr}, using the equality:
    \begin{equation}\label{eq: recursive equality rhobar}
        \begin{aligned}
            &\rhobar\big((a\otimes D)v\big)(b)\\
        &=a a_1\cdots a_n D\circ D_1\circ \cdots \circ D_n(b)\\
        &=(aD)\circ\left(a_1\cdots a_n D_1\circ \cdots \circ D_n\right)(b)- aD(a_1\cdots a_n)D_1\circ \cdots \circ D_n(b)\\
        &=\rhobar(a\otimes D)\circ\rhobar(v)(b)-\rhobar\big((a\otimes D)\tr v\big)(b).
        \end{aligned}
    \end{equation}
First, for every monomial $u\in \Bbar_{\envU(L)}$ denoting its length $\ell(u)$, we prove by induction on $n\geq 1$ the assertion:
\begin{align*}
    \P(n): "&~ \text{for every $c\in \B_L$, the following set is finite:}\\
    &\left\{(u, b)\in\Bbar_{\envU(L)} \times\B_\A,~\ell(u)\leq n: \la\rhobar(u)(b),c\ra\ne 0\right\}."
\end{align*}
 
If $n=1$, then $\P(1)$ is the Assumption \eqref{assump: finiteness rhobar}.   
Suppose that $\P(n)$ is true for a certain $n\geq 1$. For $a\otimes D\in\B_L$, $v= (a_1\otimes D_1)\ldots (a_n\otimes D_n)\in \Bbar_{\envU(L)}$, and $b,c\in\B_\A$ by equality \eqref{eq: recursive equality rhobar}:
\[
\la\rhobar\big((a\otimes D)v\big)(b),c\ra=\la\rhobar(a\otimes D)\circ\rhobar(v)(b),c\ra-\la\rhobar\big((a\otimes D)\tr v\big)(b),c\ra.
\]
Therefore:
\begin{multline*}
    \la\rhobar\big((a\otimes D)v\big)(b),c\ra\neq 0\\
    \Rightarrow~\la\rhobar(a\otimes D)\circ\rhobar(v)(b),c\ra\neq 0\quad \vee \quad \la\rhobar\big((a\otimes D)\tr v\big)(b),c\ra\neq 0.
\end{multline*}

For the first term, by definition \eqref{eq:dualityA} of the pairing $\la\cdot,\cdot\ra$, one has that: 
\[\rhobar(v)(b)=\sum_{d\in\B_A}\la \rhobar(v)(b),d \ra d.\]
Then one can write by linearity:
\[\la\rhobar(a\otimes D)\circ\rhobar(v)(b),c\ra=\sum_{d\in\B_A}\la\rhobar(v)(b),d \ra \la \rhobar(a\otimes D)(d),c\ra.\]
By $\P(1)$, there exist finitely many couples $(a\otimes D,d)\in \B_L\times \B_\A$ such that $\la \rhobar(a\otimes D)(d),c\ra$ and for every such couple, by $\P(n)$, there exist finitely many couples $(v,b)\in\Bbar_{\envU(L)} \times\B_\A,~\ell(v)=n$ such that $\la\rhobar(v)(b),d \ra\neq 0$.
We deduce that there exist finitely many $(a\otimes D, v, b)\in\B_L\times \B_{\envU(L)}\times\B_\A$, $\ell(v)=n$, such that:
\[\la\rhobar(a\otimes D)\circ\rhobar(v)(b),c\ra\neq 0.\]

For the second term, given $v=(a_1\otimes D_1)\ldots (a_n\otimes D_n)$, we have:
\begin{align*}
    (a\otimes D)\tr v&=\sum_{i=1}^n (a_1\otimes D_1)\ldots (a D(a_i)\otimes D_i)\ldots (a_n\otimes D_n)\\
    &=\sum_{i=1}^n \sum_{d\in \B_\A}\la aD(a_i),d\ra (a_1\otimes D_1)\ldots (d\otimes D_i)\ldots (a_n\otimes D_n)
\end{align*}
so that
\begin{align*}
    &\la\rhobar\big((a\otimes D)\tr v\big)(b),c\ra\\
    &=\sum_{i=1}^n \sum_{d\in \B_\A}\la aD(a_i),d\ra \la\rhobar\Big((a_1\otimes D_1)\ldots (d\otimes D_i)\ldots (a_n\otimes D_n)\Big)(b),c\ra.
\end{align*}

By $\P(n)$, for every $i=1,\ldots,n$, the set of all 
$$u=(a_1\otimes D_1) \cdots (d\otimes D_i) \cdots (a_n\otimes D_n)\in\Bbar_{\envU(L)}$$
such that $\la\rhobar(u)(b),c\ra\ne 0$ is finite,
and therefore $\P(n+1)$ follows.
\end{proof}
   
Thus we arrive at the following statement:
\begin{prop}\label{prop:trbar}
If Assumptions \ref{assump: finiteness Lie bracket AtensDer(A)} and \ref{assump: finiteness rhobar} are satisfied, then on $\envU(L)$ the coalgebra structure $(\Delta_\trbar,\varepsilon)$ dual to the algebra structure $(\trbar,\mathds 1)$ defined in Proposition \ref{prop: dual Hopf algebra} with respect to the pairing \eqref{eq:dualityL} is given by:
\begin{align}
\nonumber
\Delta_\trbar\ind&=\ind\bm{\otimes} \ind, \\
\nonumber
\Delta_\trbar(u\ast v)&=\Delta_\trbar(u)\ast\Delta_\trbar(v),
\\ \Delta_\trbar(a\otimes D)&=(a\otimes D)\bm{\otimes} \mathds 1 + \sum_{u\in\B_{\envU(L)}} Tu\bm{\otimes} (\Theta(u\otimes a)\otimes D), \label{eq:deltatrbar}
\end{align}
where we define the map $\Theta:\envU(L)\otimes\A\to\A$,
\begin{equation}\label{eq:Theta}
\Theta(u\otimes a):= \sum_{b\in\B_\A} \ind_{(b\otimes D \in \B_L)}\la \rhobar(u)(b),a\ra b .
\end{equation}
\end{prop}
\begin{proof}
First, let's remark that the coalgebra structure $(\Delta_\trbar,\varepsilon)$ is well defined: indeed Assumption \ref{assump: finiteness Lie bracket AtensDer(A)} is equivalent to Assumption \ref{assump: finiteness Lie bracket}, and by Lemma \ref{lem: implication assumptions}, Assumption \ref{assump: finiteness rhobar} implies Assumption \ref{assump: finiteness tr}, then Proposition \ref{prop: dual Hopf algebra} applies and proves that the coalgebra structure is well-defined and permits to derive the two first equalities.

Let us prove the third equality by simple computation: by equation \eqref{eq: post-Lie associative product}, for all $(a\otimes D)\in \B_L$ and all $u,v\in \envU(L)\setminus\{\mathds 1\}$, we have
\begin{equation*}
\begin{split}
    \la u\bm{\otimes} v, \Delta_\trbar(a\otimes D) \ra&=\la u\trbar v, a\otimes D \ra
    \\ &=\left\{
    \begin{array}{ll}
    \la u\tr v, a\otimes D\ra\quad &\text{if}~v=\tilde a\otimes \tilde D\in L,\\ \\
    0 &\text{else}.
    \end{array}
    \right.
\end{split}
\end{equation*}
Thus, since $u\tr(\tilde a \otimes \tilde D)=\rhobar(u)(\tilde a)\otimes \tilde D$, we obtain that: 
\begin{equation*}
    \la u\bm{\otimes} v,\Delta_\trbar(a\otimes D)\ra=\sum_{b\in\B_\A}\la \rhobar(u)(b),a\ra_{\A} \la v,b\otimes D\ra_{\envU(L)}
\end{equation*} 
which concludes the proof by \eqref{eq:Theta}.
\end{proof}

See \cite{GMZ} and \cite{BH24} for particular cases of formula \eqref{eq:deltatrbar} in the context of multi-indices 
\cite{LOT}.

\subsection{Extension of the representation map}

We fix again a basis $\B_\A$ of $\A$ and we denote by $\overline\A$ the space of formal series $\sum_{\gamma\in\B_\A} a_\gamma \gamma$, $a_\gamma\in\R$. 
The canonical pairing $\la \cdot,\cdot \ra$ on $\overline\A\times\A$ given by (the bilinear extension of) $\la\gamma,\beta\ra=\ind_{(\gamma=\beta)}$, $\gamma,\beta\in\B_\A$, allows to identify $\overline\A$ with the dual $\A^*$ by setting for all $\beta\in\B_\A$:
$$\left(\sum_{\gamma\in\B_\A} a_\gamma \gamma\right)(\beta):=\sum_{\gamma\in\B_\A} a_\gamma \la \gamma, \beta\ra.$$

In the following Proposition, given $f\in \envU(L)^*$ we define a map $\rhobar(f):\overline\A\to\overline\A$ by making an abuse of notation for simplicity.
\begin{prop}\label{prop: extension rhobar}
If Assumption \ref{assump: finiteness rhobar} is satisfied, then for all $f\in \envU(L)^*$ the map
$\rhobar(f):\overline\A\to\overline\A$ given by:
$$\rhobar(f)\left(\sum_{\gamma\in\B_\A} a_\gamma \gamma\right):=\sum_{\beta\in\B_\A} 
\left[\sum_{\gamma\in\B_\A}\sum_{u\in\B_{\envU(L)}}f(Tu) \,a_\gamma \la \rhobar(u)(\gamma),\beta\ra \right] \beta$$
is well-defined.
\end{prop}
\begin{proof} We need to prove that for all fixed $\beta\in \B_\A$, the sum inside brackets is finite, that is to say, it has a finite number of non-zero coefficients.
This is a consequence of Proposition \ref{prop: finiteness rhobar}, since for $\beta$ fixed, the set of $(u,\gamma)\in\B_{\envU(L)}\times \B_\A$ such that $\la\rhobar(u)(\gamma),\beta\ra\ne 0$ is finite.
\end{proof}

Recall the sets $G,H\subset\envU(L)^*$ from Definition \ref{def:grch}.
\begin{prop}\label{prop: rhobar group morphism}
    If Assumptions \ref{assump: finiteness Lie bracket AtensDer(A)}-\ref{assump: finiteness rhobar} are satisfied, then the map $f\mapsto \rhobar(f)$ is a group morphism from $(H,\trbar,{\mathds1}^*)$ to  $({\rm Aut}(\overline\A),\circ,\mathrm{id})$, see Proposition \ref{pr:grch}.
\end{prop}

\begin{proof} Let $\overline\A\ni\varphi= \sum_{\gamma\in\B_\A} a_\gamma \gamma$. By the definition of $\rhobar$ in Proposition \ref{prop: extension rhobar}, we have
\[
\rhobar(f_1\trbar  f_2)(\varphi)  =\sum_{\beta\in\B_\A}\left[\sum_{\gamma\in\B_\A}\sum_{u\in\B_{\envU(L)}} f_1\trbar  f_2 (Tu)\, a_\gamma \la \rhobar(u)(\gamma),\beta\ra\right] \beta.
\]
By Assumption \ref{assump: finiteness rhobar}, for every $\beta\in\B_\A$ only a finite number of terms in the sum in brackets are non-zero; now by \eqref{eq: Delta_trbar}, we have
\[
 f_1\trbar  f_2 (Tu) = \sum_{u_1,u_2\in\B_{\envU(L)}} \, f_1(Tu_1)f_2(Tu_2)\, \la u_1\trbar u_2,Tu \ra,
 \]
 and by the finiteness property of Corollary \ref{lem:wd} the latter sum contains only a finite number of non-zero terms. For each such pair $(u_1,u_2)$ 
 \[
\begin{split}
 \sum_{u\in\B_{\envU(L)}} \, \la u_1\trbar u_2,Tu \ra \, \la \rhobar(u)(\gamma),\beta\ra & = \la \rhobar(u_1\trbar u_2)(\gamma),\beta\ra \\
 &= \la\rhobar(u_1)\circ\rhobar(u_2)(\gamma),\beta\ra
\end{split} 
\]
where we have used \eqref{eq: representation algebra derivation2} in the second equality.
Therefore, again by the definition of $\rhobar$,
\[
\begin{split}
&\rhobar(f_1\trbar  f_2)(\varphi)
\\ & = \sum_{\beta\in\B_\A}\left[\sum_{\gamma\in\B_\A} \sum_{u_1,u_2\in\B_{\envU(L)}} f_1(Tu_1)f_2(Tu_2)\,a_\gamma\la\rhobar(u_1)\circ\rhobar(u_2)(\gamma),\beta\ra \right] \beta\\
&= \rhobar(f_1)\circ\rhobar(f_2)(\varphi).
\end{split}
\]
The proof is complete.
\end{proof}
In particular, we have that the map $f\mapsto \rhobar(f)$ is a group morphism from $(G,\trbar,{\mathds1}^*)$ to $({\rm Aut}(\overline\A),\circ,\mathrm{id})$, see Proposition \ref{pr:grch}.
\medskip

\subsection{Module and comodule structures}\label{sec:modcomod}

By definition, $(\A,\Delta)$ is a left $(\envU(L),\Delta_\trbar)$-comodule if $\Delta:\A\to \envU(L)\otimes \A$ satisfies
$$(\id\otimes\Delta)\Delta=(\Delta_\trbar\otimes\id)\Delta.$$

\begin{prop}\label{prop: coaction}
We suppose that Assumptions \ref{assump: finiteness Lie bracket AtensDer(A)}-\ref{assump: finiteness rhobar} are satisfied.
Let the map $\Delta:\A\to \envU(L)\otimes \A$ be defined by
\[
\Delta a = \sum_{\substack{u\in \B_{\envU(L)}\\ b\in\B_\A}} \la \rhobar(u)(b),a\ra Tu\otimes b
\]
Then $(\A,\Delta)$ is a left $(\envU(L),\Delta_\trbar)$-comodule.
\end{prop}
\begin{proof} We apply twice the definition of $\Delta$ to obtain
\[(\id\otimes\Delta)\Delta a=\sum_{\substack{u_1,u_2\in \B_{\envU(L)}\\b_1,b_2\in\B_\A}} \la \rhobar(u_1)(b_1),a\ra \la\rhobar(u_2)(b_2),b_1\ra Tu_1\otimes Tu_2\otimes b_2.\]

Now, under Assumption \ref{assump: finiteness rhobar}, by Proposition \ref{prop: finiteness rhobar}, 
there is only a finite number of non-zero terms the latter sum. Now
\[
\begin{split}
 \sum_{b_1\in\B_\A} &\la \rhobar(u_1)(b_1),a\ra  \la\rhobar(u_2)(b_2),b_1\ra \\
 &=  \la \rhobar(u_1)\circ \rhobar(u_2) (b_2),a\ra  
\\ & = \la \rhobar(u_1\star u_2)(b_2),a\ra \\
&= \sum_{u\in \B_{\envU(L)}} \la\rhobar(u)(b),a\ra \la u_1\trbar  u_2,Tu\ra Tu_1\otimes Tu_2.
\end{split}
\]
Therefore
\[
\begin{split}
(\id\otimes\Delta)\Delta a & =
\sum_{\substack{u,u_1,u_2\in \B_{\envU(L)}\\b\in\B_\A}}\la \rhobar(u)(b_2),a\ra \la u_1\trbar  u_2,Tu\ra Tu_1\otimes Tu_2\otimes b
\\ & =
(\Delta_\trbar\otimes\id)\Delta a,
\end{split}
\]
where in the last equality we have used \eqref{eq: Delta_trbar} and the fact that in the latter sum only a finite number of terms are non-zero by Proposition \ref{prop: finiteness rhobar} and by the finiteness property of Corollary \ref{lem:wd}.\\
Moreover, recalling the counit map $\varepsilon: \env_{[\cdot,\cdot]}(L)\to\R$ defined in Section \ref{sec:ast}, we have
\[(\varepsilon\otimes \id)\Delta a=\sum_{\substack{u\in \B_{\envU(L)}\\b\in\B_\A}} \la \rhobar(u)(b),a\ra \varepsilon(Tu)\otimes b=\sum_{b\in\B_\A}\la \rhobar(\ind)(b),a\ra b=a,\]
since $\rhobar(\ind)(b)=b$, where we have identified $\R\otimes \A$ with $\A$ in the two last equalities.\\
This completes the proof.
\end{proof}
\bigskip

\section{Derivations on multi-indices}\label{sec:LOT}

We want here to give an application of the results of the previous sections to an algebraic structure which has been
unveiled recently in \cite{LOT}, with applications to stochastic Taylor developments of solutions to SPDEs.

We note $\N=\{0,1,\ldots\}$ and given an integer $d\geq 1$, we use the following notations: 
\[
\N^d_*:=\N^d\setminus\{\0\}, \qquad \0:=(0,\ldots,0)\in\N^d.
\]
Then we define $\M$ as the set of compactly supported $\gamma:\N\sqcup\N^d_*\to\N$, namely $\gamma_i\ne0$ only for finitely many $i\in\N\sqcup\N^d_*$. Elements of $\M$ are called \textit{multi-indices}.  
Note that $\M$ is stable under addition: if $\gamma^1,\gamma^2\in\M$ then 
\begin{equation}\label{eq:sumga}
\gamma_i:=\gamma^1(i)+\gamma^2(i), \qquad i\in\N\sqcup\N^d_*,
\end{equation}
defines a new element in $\M$. It is also
possible to define the difference $\gamma^1-\gamma^2\in\M$ if $\gamma^1\geq\gamma^2$.
\medskip

\subsection{The Linares-Otto-Tempelmayr (LOT) setting}
In \cite{LOT} the authors developed a new tree-free approach to regularity structures.
In this subsection we start to introduce some of their main definitions. 
Let us consider the polynomial algebra 
\begin{equation*}
\A:=\R[\z_k,\z_\n]_{k\in\N,\n\in\N^d_*}
\end{equation*}
where $\{\z_k,\z_\n: k\in\N,\n\in\N^d_*\}$ are commuting variables and $\1\in\A$ is the unit. A canonical basis  for $\A$ is given by the set $\{\z^\gamma: \gamma\in \M\}$, where
$$\z^\gamma:=\prod_{i\in\N\sqcup\N^d_*}\z_i^{\gamma_i}, \quad \gamma\in\M,
\qquad \z^0=\1.$$
Then the sum in $\M$ defined in \eqref{eq:sumga} allows to describe the product in $\A$
\[
\z^{\gamma}\z^{\gamma'}=\z^{\gamma+\gamma'}, \qquad \gamma,\gamma'\in\M.
\]

Two sets of derivations on $\A$ are of interest here (see \cite[(3.9) and (3.12)]{LOT})
\begin{enumerate}
    \item The \textit{tilt} derivations $\{D^{(\n)}\}_{\n\in\N^d}$, defined by:
    \begin{equation}\label{eq: tilt}
        D^{(\0)}:=\sum_{k\ge 0}(k+1)\z_{k+1}\partial_{\z_k}\qquad \text{and}\qquad D^{(\n)}:=\partial_{\z_\n},~ \text{for}~ \n\in\N^d_*.
    \end{equation}
    \item The \textit{shift} derivations $\partial_i$, defined for $i\in\{1,\ldots,d\}$ by:
    \begin{equation}\label{eq: shift}
        \partial_i:=\sum_{\n\in\N^d}(n_i+1)\z_{\n+\e_i}D^{(\n)}
    \end{equation}
    where $\e_1=(1,0,\ldots,0)$, $\e_2=(0,1,\ldots,0)$, etc.
\end{enumerate}

For $k\in\N$ we denote by $e_k\in\M$ the multi-index $e_k(i)=\ind_{(i=k)}$ for $i\in\N\sqcup\N^d_*$, and
similarly $e_\n\in\M$ for $\n\in\N^d_*$. Explicit computations for all $\gamma\in\M$, $\n\in\N^d_*$ 
and $i\in\{1,\ldots,d\}$ show that for the tilt derivations
\begin{align}\label{eq:D0}
    D^{(\0)}\z^\gamma&=\sum_{k\geq 0}(k+1)\gamma_k\z^{\gamma+e_{k+1}-e_k}, \\
    D^{(\n)}\z^\gamma&=\gamma_\n\,\z^{\gamma-e_\n}\qquad \text{if}~\n\in\N^d_*, \nonumber
\end{align}
while for the shift derivations
\begin{align}\label{eq:partial_i}
    \partial_i\z^\gamma&=\sum_{\n\in\N^d}(n_i+1)\z_{\n+\e_i}D^{(\n)}\z^\gamma\\
    &=\sum_{k\geq 0}(k+1)\gamma_k\z^{\gamma+e_{k+1}-e_k+e_{\e_i}}+\sum_{\n\in\N^d_*}(n_i+1)\gamma_\n\z^{\gamma-e_\n+e_{\n+\e_i}}.
\nonumber
\end{align}
While $D^{(\0)}$ and $\partial_i$ are defined by infinite series, for each $\gamma\in\M$ the sums in \eqref{eq:D0}-\eqref{eq:partial_i} are finite because $\gamma$ has compact support.

The authors in \cite[\S 3.8]{LOT} used a geometrical point of view to define a binary operation denoted $\tr$ which corresponds to the covariant derivative of vector fields on the infinite dimensional manifold $\R[\z_k,\z_\n]_{k\in\N,\n\in\N^d_*}$ whose geometry is given by the canonical flat and torsion free connexion. However, this natural approach turns out to be difficult to handle because of the non-stability of the space $L_{\rm LOT}:=\R\{\partial_i\}_i\oplus\R\{\z^\gamma D^{(\n)}\}_{\gamma,\n}$ under the binary operation $\tr$. For example the covariant derivatives $\partial_i\tr \partial_i$ cannot be expressed as a linear combination of the aforementioned derivations and thus does not belong to $L_{\rm LOT}$.
\medskip

\subsection{Post-Lie algebra structure}
In order to use the results of the preceeding sections, we redefine the space $L_{\rm LOT}$ of \cite{LOT} in a different manner. Denoting again $\A:=\R[\z_k,\z_\n]_{k\in\N,\n\in\N^d_*}$, we define the space $L_0$ as the subspace of $\A\otimes \Der(\A)$ generated by the elements $\{\z^\gamma \otimes D^{(\n)}\}_{\n\in\N^d,\gamma\in\M}$ and $ \{\1\otimes \partial_i\}_{i\in\{1,\ldots,d\}}$, namely:
\begin{equation}\label{eq: post-Lie algebra L}
L_0:=\text{Span}\{\1\otimes\partial_i\}_{i\in\{1,\ldots,d\}} \oplus \text{Span}\left\{\z^\gamma \otimes D^{(\n)}\right\}_{\gamma\in\M,\n\in\N^d},
\end{equation}
where $\1$ is the unit in $\A$.

\begin{theorem}\label{thm:postLOT} Setting $\A:=\R[\z_k,\z_\n]_{k\in\N,\n\in\N^d_*}$, 
 the space $L_0$ is a sub-post-Lie algebra of $\A\otimes \Der(\A)$, for the canonical post-Lie algebra structure $(\tr,[\cdot,\cdot])$ given in Theorem \ref{theo: post-Lie structure from derivations}.
\end{theorem}
\begin{proof}
Let us verify that $L_0$ is stable under the action of the post-Lie structure $(\tr,[\cdot,\cdot])$ induced by $\A\otimes \Der(\A)$ given in Theorem \ref{theo: post-Lie structure from derivations}, namely, for $a_1\otimes D_1,a_2 \otimes D_2\in L_0$ we have
\begin{align*}
    (a_1\otimes D_1)\tr (a_2 \otimes D_2)&=a_1 D_1(a_2)\otimes D_2\in L_0,\\
    [a_1\otimes D_1, a_2 \otimes D_2]&=a_1a_2\otimes [D_1,D_2]_\circ\in L_0.
\end{align*}
By definition of the operation $\tr$, we obtain the following equalities:
\begin{align}
    (\1\otimes \partial_i)\tr (\1\otimes \partial_j) &=  \partial_i(\1)\otimes \partial_j=0, \label{eq:partial12}
    \\   (\z^\gamma \otimes D^{(\n)}) \tr (\1\otimes \partial_i)  &=\z^\gamma D^{(\n)}(\1)\otimes \partial_i=0, \label{eq:zgaDntr} \\
 (\z^{\gamma'}\otimes D^{(\n')}) \tr (\z^\gamma \otimes D^{(\n)}) &=\z^{\gamma'} D^{(\n')}\z^{\gamma}\otimes D^{(\n)}\in L_0,\\
 (\1\otimes \partial_i) \tr (\z^\gamma \otimes D^{(\n)}) &= \partial_i\z^\gamma\otimes D^{(\n)}\in L_0.
\end{align} 
where in \eqref{eq:partial12} and \eqref{eq:zgaDntr}, we used the property that derivations vanish once evaluated at $\1\in\R[\z_k,\z_\n]_{k\in\N,\n\in\N^d_*}$.
 
It remains to discuss the bracket.  Let us first compute the Lie bracket $[\cdot,\cdot]_\circ$ on the family of derivations $\{D^{(\n)},\partial_i\}_{\n,i}$:
\begin{enumerate}[leftmargin=*]
    \item By the definitions, for all $\n,\n'\in\N^d$ the derivations $D^{(\n)}$ and $D^{(\n')}$ commute, i.e. $[D^{(\n)},D^{(\n')}]_\circ=0$.
    \item  For all $\{i,j\}\in\{1,\ldots,d\}$, we have
    \begin{align*}
        \partial_i\circ\partial_j=&\sum_{\n\in\N^d}(n_i+1)(n_j+1)\z_{\n+\e_i+\e_j}D^{(\n)}\\
        & +\sum_{\n,\m\in\N^d}(n_i+1)(m_j+1)\z_{\n+\e_i}\z_{\m+\e_j}D^{(\n)}D^{(\m)}.
    \end{align*}
Since this is symmetric in $(i,j)$, we have $[\partial_i,\partial_j]_\circ=0$.
    \item Since for all $\n\in\N^d$ the derivations $D^{(\0)}$ and $D^{(\n)}$ commute, one has:
    $$D^{(\0)}\circ \partial_i=\sum_{\n\in\N^d}(n_i+1)\z_{\n+\e_i}D^{(\0)}\circ D^{(\n)}=\partial_i\circ D^{(\0)}.$$
    Moreover for all $\n=(n_1,\ldots,n_d)\in\N^d_*$, we have
    $$D^{(\n)}\circ \partial_i=n_iD^{(\n-\e_i)}+\partial_i\circ D^{(\n)}.$$
    Thus for all $\n\in\N^d$:
    \begin{equation*}
        [D^{(\n)}, \partial_i]_\circ=n_iD^{(\n-\e_i)}.
    \end{equation*}
\end{enumerate}
In conclusion, we have for all $\n,\n'\in\N^d$ and $\gamma,\gamma'\in \M$:
\begin{align}
\label{eq:zgazga}    [\z^\gamma\otimes D^{(\n)},\z^{\gamma'}\otimes D^{(\n')}]&=0  \\
\label{eq:paipaj}     [\1\otimes \partial_i,\1\otimes \partial_j]&=0  \\
\label{eq:zgaDn}    [\z^\gamma\otimes D^{(\n)},\1\otimes\partial_i]&=n_i(\z^\gamma\otimes D^{(\n-\e_i)})\in L_0.
\end{align}
The proof is complete.
\end{proof}

Let us now compute on the previously defined basis of $L_0$, the Lie bracket $\db{\cdot,\cdot}$ given by the relation \eqref{eq: composition Lie bracket}:
\begin{align}
\label{eq:Lie1}    \db{\z^\gamma\otimes D^{(\n)},\z^{\gamma'}\otimes D^{(\n')}}&=\z^\gamma D^{(\n)}\z^{\gamma'}\otimes D^{(\n')}-\z^{\gamma'} D^{(\n')}\z^\gamma\otimes D^{(\n)},\\
\label{eq:Lie2}    \db{\z^\gamma \otimes D^{(\n)} , \1\otimes \partial_i}&=n_i\,\z^\gamma\otimes D^{(\n-\e_i)}-\partial_i\z^\gamma\otimes D^{(\n)},\\
\label{eq:Lie3}    \db{\1\otimes \partial_i, \1\otimes \partial_j}&=0.
\end{align}

\begin{remark}{\rm
In \cite[formula (3.36)]{LOT} we find the formula
\[
\z^\gamma D^{(\n)}\tr \partial_i = n_i \, \z^\gamma D^{(\n-\e_i)}.
\]
This differs from our \eqref{eq:zgaDntr}. Moreover in \cite{LOT} the operator $\partial_1 \triangleleft \partial_2$ can
not be written as a finite linear combination of $\{\partial_i\}\cup\{z^\gamma D^{(\n)}\}_{\gamma,\n}$, while in our setting we have the
simple expression \eqref{eq:partial12}. Therefore the post-Lie algebra we define is different from the (partial) pre-Lie algebra constructed on the space $L_{\rm LOT}=\R\{\partial_i\}_i\oplus\R\{\z^\gamma D^{(\n)}\}_{\gamma,\n}$ in \cite{LOT}.

However, the Lie algebra defined by $\db{\cdot,\cdot}$ is compatible with the Lie algebra $[\cdot,\cdot]_\circ$ in \cite[\S 3.10]{LOT}. Indeed, 
the relations \eqref{eq:Lie1}-\eqref{eq:Lie2}-\eqref{eq:Lie3} show that the Lie-algebra morphism $\rhohat:(\env_{\db{\cdot,\cdot}}(L_0),\db{\cdot,\cdot})\to (\Der(\A),[\cdot,\cdot]_\circ)$ of \eqref{eq:rhohat} allows to recover the structure described
in \cite[\S 3.10]{LOT}, see in particular \cite[(3.46)-(3.47)]{LOT}.

On the other hand the post-Lie algebra we define is isomorphic via $\rhobar$ to the one written in \cite[Theorem 5.5]{bruned2022post}.
Our construction has the merit of being more general and to distinguish the abstract enveloping algebra $\envU(L)$ from its realisation as an
algebra of endomorphisms of $\A$.

}\end{remark}
    
    \begin{remark}\label{rem: finiteness Lie bracket AtensDer(A)}{\rm
    By the equalities \eqref{eq:zgazga} \eqref{eq:paipaj} and \eqref{eq:zgaDn}, Assumption \ref{assump: finiteness Lie bracket AtensDer(A)} is trivially satisfied in this setting.
}\end{remark}
\medskip

    \subsection{A basis for the enveloping algebra}\label{subsec: basis env alg}
    The isomorphism of Theorem \ref{theo: isomorphism between enveloping algebras} allows us to work with the space $\env_{[\cdot,\cdot]}(L_0)$ for which the multiplication table of the associative product $\trbar$ can be written explicitly, once one fixes a basis. In this section
    we recover the basis \cite[(4.15)]{LOT}, see \eqref{eq: rhobar(E_mF_J)} below. 
    
    The Poincaré-Birkhoff-Witt Theorem \ref{theo: PBW}, permits us to exhibit a choice of basis for $\env_{[\cdot,\cdot]}(L_0)$ which depends on an ordering of the basis of $L_0$ given by the derivations of type $\z^\gamma \otimes D^{(\n)}$ and $\1\otimes \partial_i$. 
    
    The commutation relations \eqref{eq:zgazga}-\eqref{eq:paipaj}-\eqref{eq:zgaDn} indicate that in order to apply the PBW theorem, we only need to choose an order between the elements of type $\1\otimes \partial_i$ and of type $\z^\gamma \otimes D^{(\n)}$. In particular if we choose that $\1\otimes\partial_i<\z^\gamma \otimes D^{(\n)}$ for all $i\in\{1,\ldots,d\},\n\in\N^d,\gamma\in\M$ one obtains the following basis for $\env_{[\cdot,\cdot]}(L_0)$ given by the set of equivalence classes of monomials of the form
    \begin{equation}\label{eq:basis1}
        (\1\otimes\partial_1)^{m_1}\ldots (\1\otimes\partial_d)^{m_d}(\z^{\gamma_1}\otimes D^{(\n_1)} )\ldots (\z^{\gamma_k}\otimes D^{(\n_k)})
    \end{equation}
    where $(m_1,\ldots,m_d)\in \N^d$ and $(\gamma_l, \n_l)\in \M\times \N^d$ for all $l\in\{1,\ldots,k\}$. 

        From the commutation relation \eqref{eq:zgaDn} we deduce that we can write any monomial in $\env_{[\cdot,\cdot]}(L_0)$ in the form \eqref{eq:basis1}:
    \begin{equation}\label{eq: inversion order basis}
    \begin{split}
        &(\z^{\gamma_1}\otimes D^{(\n_1)})\ldots (\z^{\gamma_k}\otimes D^{(\n_k)})(\1\otimes\partial_1)^{m_1}\ldots (\1\otimes\partial_d)^{m_d}\\
        &=(\1\otimes\partial_1)^{m_1}\ldots (\1\otimes\partial_d)^{m_d}(\z^{\gamma_1}\otimes D^{(\n_1)} )\ldots (\z^{\gamma_k}\otimes D^{(\n_k)})\\
        &+\ind_{(|\m|>0)} \sum_{\bar\n_1,\ldots,\bar\n_k} \ind_{\left(\sum_j|\n_j-\bar\n_j|=|\m|\right)} \prod_{l=1}^k \left[\ind_{(\bar\n_l\le\n_l)}
        \frac{\n_l!}{\bar\n_l!} (\z^{\gamma_l}\otimes D^{(\bar\n_l)})\right],
 \end{split}
 \end{equation}
  where $|\m|:=m_1+\ldots+m_d$  and $\bar\n\le\n\Longleftrightarrow (\bar n_1\le n_1)\wedge \ldots \wedge (\bar n_d\le n_d)$, and  we use standard notations for $\n=(n_1,\ldots,n_d),\m=(m_1,\ldots,m_d)\in\N^d$:
    \[
    \n!=n_1!\cdots n_d!, \qquad \binom{\n}{\m}=\binom{n_1}{m_1}\cdots \binom{n_d}{m_d}.
    \]
Therefore we consider \eqref{eq:basis1} as a {\it normal ordering}
of monomials in $\env_{[\cdot,\cdot]}(L_0)$.

    We denote for $\m\in\N^d$
    \[
    (1\otimes\partial)^{\m}:=(\1\otimes\partial_1)^{m_1}\ldots (\1\otimes\partial_d)^{m_d}.
    \]
        In order to choose a normalisation for the basis element in \eqref{eq:basis1}, we note that the commutation relations 
    \eqref{eq:zgazga}-\eqref{eq:paipaj}
    \[[\z^\gamma\otimes D^{(\n)},\z^{\gamma'}\otimes D^{(\n')}]=[1\otimes\partial_i,1\otimes\partial_j]=0\] imply that $\R\{\z^\gamma\otimes D^{(\n)}\}_{\gamma\in\M,\n\in\N^d}$ and $\R\{\1\otimes\partial_i\}_{i\in\{1,\ldots,d\}}$ are commutative subalgebras of $\env_{[\cdot,\cdot]}(L_0)$. In particular the coshuffle coproduct $\Cop$ defined in \eqref{eq:coshuffle} acts
on these two algebras as follows 
\begin{align*}
    \Cop(1\otimes\partial)^{\m} &=\sum_{\m'+\m''=\m}\binom{\m}{\m'}\, (1\otimes\partial)^{\m'}{\otimes}(1\otimes\partial)^{\m''}, \\
    \Cop a^\ell &=\sum_{k=0}^\ell \binom{\ell}{k}\, a^k {\otimes}\,
    a^{\ell-k}, \qquad a= \z^\gamma\otimes D^{(\n)}, \\
    \Cop\prod_{i=1}^na_i^{\ell_i} &=\prod_{i=1}^n\Cop a_i^{\ell_i}, \qquad 
    a_i = \z^{\gamma_i}\otimes D^{(\n_i)}, \quad a_i\ne a_j \ \text{if} \ i\ne j.
\end{align*}
Therefore we choose a normalisation which allows to minimise the combinatorial coefficients in these expressions. 
For $\m=(m_1,\ldots,m_d)\in\N^d$ we set 
     $$E_{\m}:=\frac{1}{\m!}(1\otimes\partial)^{\m},
\qquad \m!:=m_1!\cdots m_d!.
$$
We define now multi-indices $J$ on $\M\times\N^d$, namely functions $J:\M\times\N^d\to\N$ with finite support, and we define
$$F_{J}:=\prod_{(\gamma,\n)\in \M\times\N^d} \frac{1}{J(\gamma,\n)!}(\z^{\gamma}\otimes D^{(\n)} )^{J(\gamma,\n)}.$$
We use the convention $E_\0=F_{\emptyset}=\mathds{1}\in\env_{[\cdot,\cdot]}(L_0)$. 
These definitions allow to express the coalgebra structure of $\env_{[\cdot,\cdot]}(L_0)$ given by the coshuffle coproduct $\Cop$ defined in \eqref{eq:coshuffle}, given on such elements by
\begin{gather}
\nonumber    \Cop(E_{\m})=\sum_{\m'+\m''=\m}E_{\m'}{\otimes}E_{\m''}, \qquad
    \Cop(F_{J})=\sum_{J'+ J''=J}F_{J'}{\otimes}F_{J''}, \\
\label{eq:coco}    \Cop(E_{\m}F_{J})=\Cop(E_{\m})\Cop(F_{J})=\sum_{\substack{\m'+\m''=\m \\ J'+ J''=J}}E_{\m'}F_{J'}\otimes E_{\m''}F_{J''},
\end{gather}
which is \eqref{eq:copast} in this setting. In addition, for 
the concatenation product in $\env_{[\cdot,\cdot]}(L_0)$ we obtain
\begin{equation}\label{eq:star}
    E_{\m}E_{\bar\m}=\binom{\m+\bar\m}{\m}\, E_{\m+\bar\m}\qquad \text{and} \qquad F_{J}F_{\bar J}=
    \binom{J+\bar J}{J}\, F_{J+ \bar J},
\end{equation}
where the binomial coefficient is given by
\[
\binom{J+\bar J}{J}:=\frac{(J+\bar J)!}{J!\,\bar J!}, \qquad 
J!:=\prod_{(\gamma,\n)}J(\gamma,\n)!.
\]

We denote $\E=\{E_\m\}_\m$ and $\F=\{\F_J\}_J$. Then the basis $\B_{\env_{[\cdot,\cdot]}(L_0)}$ for ${\env_{[\cdot,\cdot]}(L_0)}$ defined in \eqref{eq:basis} can be
described as set of all concatenation products of type $E_{\m}F_{J}$
\begin{equation*}
        \B_{\env_{[\cdot,\cdot]}(L_0)}= \mathcal{E}\cdot\mathcal{F}=\{E_\m F_J\}_{\m,J}.
\end{equation*}
Note that this choice of a basis corresponds (via the representation $\bar\rho$) to the one that has been adopted in \cite[formula (4.15)]{LOT}: following Theorem \ref{theo: canonical representation}, the representation $\rhobar:\env_{[\cdot,\cdot]}(L_0)\to\End(\A)$ is given on basis elements $E_{\m}F_{J}\in\B_{\env_{[\cdot,\cdot]}(L_0)}$ by:
\begin{equation}\label{eq: rhobar(E_mF_J)}
     \rhobar\left(E_{\m}F_{J}\right)=\left(\frac{1}{\m!} \, \prod_{\gamma,\n}\frac{(\z^{\gamma})^{J(\gamma,\n)}}{J(\gamma,\n)!}\right)
    \partial^{\m}\circ \prod_{\gamma,\n}\left(D^{(\n)}\right)^{\circ J(\gamma,\n)}
\end{equation}
where we denote for all $\m=(m_1,\ldots,m_d)\in\N^d$
    \[
    \partial^{\m}:=(\1\otimes\partial_1)^{m_1}\ldots (\1\otimes\partial_d)^{m_d}.
    \]
    
\medskip

\subsection{An explicit formula for the product}\label{sec:expprod} 
With the notation introduced in the previous subsection, the equality \eqref{eq: inversion order basis} can be written in a more compact form
\begin{equation}\label{eq: inversion order basis simpler form} 
   F_J E_\m = E_\m F_J + \frac{\ind_{(|\m|>0)}}{\m!} \sum_{J_0\in J_\m} \frac{J_0!}{J!}
   \left(\prod_{\gamma,\n} \,(\n!)^{J(\gamma,\n)-J_0(\gamma,\n)}\right) F_{J_0}, 
\end{equation}
where for
\[
F_{J}=\frac1{J!}\prod_{i=1}^{k} (\z^{\gamma_i}\otimes D^{(\n_i)}) ,
\]
we define $J_\m$ as the set of $J_0: \M\times\N^d\to\N$ with finite support and such that there exist
$\overline\n_1,\ldots,\overline\n_k\in\N^d$ with $\overline\n_i\le\n_i$ and $\sum_{i=1}^k|\n_i-\overline\n_i|=|\m|$ such that
\[
F_{J_0}=\frac1{J_0!}\prod_{i=1}^{k} (\z^{\gamma_i}\otimes D^{(\overline\n_i)}).
\]

We want now to exhibit the extension of the post-Lie product $\tr$ on $\env_{[\cdot,\cdot]}(L_0)$ and its related associative product $\trbar$ defined in Proposition \ref{prop: associative product}, using the simplifications arising from our previous computations. Let us start by doing some simplifications for the extension of $\tr$ on $\env_{[\cdot,\cdot]}(L_0)$.

First of all, denoting as before by $\rhobar$ the representation morphism of Theorem \ref{theo: canonical representation}, Proposition \ref{prop: grafting a word on a letter derivations} implies that for $E_{\m}F_{J}\neq \mathds{1}$
\begin{equation*}
    E_{\m}F_{J}\tr(\1\otimes \partial_i)=\rhobar\left(E_{\m}F_{J}\right)(\1) \otimes \partial_i=0,
\end{equation*}
so that by point 1 in Proposition \ref{prop: extension post-lie product}
\begin{equation*}
    E_{\m}F_{J}\tr E_{\bar\m}
    =\left\{
    \begin{array}{ll}
         E_{\bar\m} & \mbox{if}~ E_{\m}F_{J} = \mathds{1}, \\
       0 & \mbox{else}.
    \end{array}
\right.
\end{equation*}
By \eqref{eq: representation algebra derivation}
\begin{equation*}
    E_{\m}F_{J}\tr(\z^{\bar\gamma}\otimes D^{(\bar\n)})= \rhobar\left(E_{\m}F_{J}\right)(\z^{\bar\gamma}) \otimes D^{(\bar\n)},
    \end{equation*}
Now from \eqref{eq: grafting word on word} and \eqref{eq:coco}
\begin{equation}\label{eq:EFF}
\begin{split}
    E_{\m}F_{J}\tr F_{\bar J} &=  E_{\m}F_{J}\tr\frac1{\bar J!}\prod_{l=1}^{N} \z^{\bar\gamma_l}\otimes D^{(\bar\n_l)}    \\
    &= \frac1{\bar J!} \sum_{\substack{\m_{1}+\cdots +\m_{N}=\m\\ J_{1}+\cdots+ J_{N}=J}}
    \prod_{l=1}^{N} 
    \left(\rhobar(E_{\m_{l}}F_{J_{l}})( \z^{\bar\gamma_l})\otimes D^{(\bar\n_l)}\right).
\end{split}
\end{equation}
Finally, using point 3 of Proposition \ref{prop: extension post-lie product} and \eqref{eq:coco}, we obtain
\begin{equation*}
    E_{\m}F_{J}\tr E_{\bar\m}F_{\bar J}= \sum_{\substack{\m'+\m''=\m \\ J'+ J''=J}} (E_{\m'}F_{J'}\tr E_{\bar\m}) (E_{\m''}F_{J''} \tr F_{\bar J}).
\end{equation*}
The only non-zero term in the sum is given for $\m''=\m$ and $J''=J$ and in that case $E_{\m'}F_{J'}=\mathds{1}$ and $\mathds{1}\tr E_{\bar\m}=E_{\bar\m}$, then:
\begin{equation*}
    E_{\m}F_{J}\tr E_{\bar\m}F_{\bar J}=E_{\bar\m}\left(E_{\m}F_{J}\tr F_{\bar J}\right).
\end{equation*}
Thus, from the Definition of $\trbar$ given in \ref{prop: associative product}, one gets:
\begin{align*}
    E_{\m}F_{J}\trbar E_{\bar\m}F_{\bar J}&=\sum_{\substack{\m'+\m''=\m \\ J'+ J''=J}}E_{\m'}F_{J'}\left(E_{\m''}F_{J''}\tr  E_{\bar\m}F_{\bar J}\right)\\
    &=\sum_{\substack{\m'+\m''=\m \\ J'+ J''=J}}E_{\m'}F_{J'}E_{\bar\m}\left(E_{\m''}F_{J''}\tr F_{\bar J}\right).
\end{align*}
Using \eqref{eq: inversion order basis simpler form}, one deduces an expression for the product $\trbar$ on the basis $\B_{\envU(L_0)}$ 
\begin{multline*}
    E_{\m}F_{J}\trbar E_{\bar\m}F_{\bar J}=\sum_{\substack{\m'+\m''=\m \\ J'+ J''=J}}\underbrace{E_{\m'}E_{\bar\m}}_{\in \E}\underbrace{F_{J'}(E_{\m''}F_{J''}\tr F_{\bar J})}_{\in\F}
    \\
    +\sum_{\substack{\m'+\m''=\m \\ J'+ J''=J}} 
    \frac{\ind_{(|\m|>0)}}{\m!} \sum_{J_0\in J'_{\bar\m}} \frac{J_0!}{J'!}
   \left[\prod_{\gamma,\n} \,(\n!)^{(J'-J_0)(\gamma,\n)}\right]
  \underbrace{E_{\m'}}_{\in\E}\underbrace{F_{J_0}(E_{\m''}F_{J''}\tr F_{\bar J})}_{\in\F}
\end{multline*}
where the elements $E_{\m''}F_{J''}\tr F_{\bar J}$ are given by \eqref{eq:EFF}.
\bigskip

\section{The structure group}\label{sec:ansatz}
In this section, we start the construction which allows to apply the results of the previous section to a specific stochastic PDE. We first explain
very briefly the main motivation of this construction, referring to \cite[\S 7]{LOT} for more details.
We choose an equation on $\R^d$ of the form
\begin{equation*}
    \L u = a(u(x))\xi,
\end{equation*}
where $\L$ is a linear differential operator which admits a Green kernel $K$, $\xi:\R^d\to\R$ (the noise term) is a fixed continuous function, $a:\R\to\R$ is smooth, and solutions are functions $u:\R^d\to\R$.  The multi-index symmetry factor is given for all $\beta\in\M$ by:
\begin{equation*}
    \sigma(\beta):=\prod_{k\in\N} (k!)^{\beta_k}.
\end{equation*}

The analytical theory of \eqref{eq:spde} is based on the following Ansatz:
any solution $u$ satisfies a local Taylor development at order $\delta>0$ of the form:
\begin{equation}\label{eq: taylor}
    u(y)=\sum_{|\beta|<\delta} \frac{1}{\sigma(\beta)}\Upsilon^{a,u}\z^\beta(x) \, \Pi_x \z^\beta(y) + R_x^\delta(y),
\end{equation}
where 
\begin{itemize}[leftmargin=*]
\item $R^\delta$ is a remainder of order $\delta$: $|R^\delta_x(y)|\lesssim|y-x|^\delta$;
\item $|\beta|\in\R^+$ is the \textit{homogeneity} of $\beta\in\M$ that is defined in our case in \eqref{eq:|} below;
\item $\{\Pi_x\z^\beta\}_{\beta\in\M}$ is a fixed family of functions 
which depend on the noise term $\xi$ and also on the Green kernel $K$;
\item $\Upsilon^{a,u}:\R[\z_k,\z_\n]_{k\in\N,\n\in\N^d_*}\to C(\R^d)$ is an explicit function depending on $a$ and $u$,
which is defined in \cite[(7.22)]{LOT}.
\end{itemize}

The functions $\{\Pi_x\z^\beta\}_{\beta\in\M}$ come with a family of linear operators $\Gamma_{xy}:\A\to\A$ such that 
\[
\Pi_x \Gamma_{xy} = \Pi_y, \qquad \forall \, x,y\in\R^d.
\]
These operators are constructed via the representation of a group $(G,\trbar)$, called the {\it structure group} of the equation.
In the rest of this paper, see in particular the final section \ref{sec:group}, we show how to construct this group with such a representation,
using the material of the previous sections.
\medskip

\subsection{Homogeneity}
Now we consider in particular the equation on $\R^d$:
\begin{equation}\label{eq:spde}
    -\Delta u = a(u(x))\xi,
\end{equation}
where $\Delta$ denotes the $d$-dimensional Laplacian operator: \[\Delta u= \frac{d^2}{dx_1} u +\ldots+ \frac{d^2}{dx_d} u.\]

We fix $\alpha\in\,]0,1[$ and we note $|\n|=|(n_1,\ldots,n_d)|=n_1+\ldots+n_d$ for $\n\in\N^d$. 
The value $\alpha\in\,]0,1[$ indicates that one expects in the non-smooth setting that $\xi$ is a distribution
in some Besov space $C^{\alpha-2}$ and $u$ is a Hölder function in $C^\alpha$.

We define the \emph{homogeneity} $|\cdot|:\M\to[0,+\infty)$ as follows:
\begin{equation}\label{eq:|}
|\beta|:=\alpha\sum_{k\geq 0}\beta_k+\sum_{\n\ne \0}|\n|\beta_\n.
\end{equation}
The homogeneity plays a crucial role since it is the expected "regularity" of the terms $\Pi_x$ in \eqref{eq: taylor}.
In particular $\Pi_x\z^\beta$ is expected to satisfy 
\[
|\Pi_x\z^\beta(y)|\lesssim |y-x|^{|\beta|}, \qquad x,y\in\R^d.
\] 

We recall the definition \eqref{eq: post-Lie algebra L} of $L$ and we define the subspace $L\subset L_0$
\begin{equation}\label{eq: sub-post-Lie algebra L}
L:=\text{Span}\{\1\otimes\partial_i\}_{i\in\{1,\ldots,d\}} \oplus \text{Span}\left\{\z^\gamma \otimes D^{(\n)}\right\}_{\gamma\in\M,\, \n\in\N^d, |\gamma|>|\n|}.
\end{equation}
where the condition $|\gamma|>|\n|$ on the elements $\z^\gamma \otimes D^{(\n)}$ will ensure the key finiteness property of Proposition \ref{pr:finite}.\\

Now we have the analog of Theorem \ref{thm:postLOT}:
\begin{theorem}\label{thm:postLOT'} Setting $\A:=\R[\z_k,\z_\n]_{k\in\N,\n\in\N^d_*}$, 
 the space $L$ is a sub post-Lie algebra of $\A\otimes \Der(\A)$, for the canonical post-Lie algebra structure $(\tr,[\cdot,\cdot])$ given in Theorem \ref{theo: post-Lie structure from derivations}.
\end{theorem}
\begin{proof}
Let us verify that $L$ is stable under the action of the post-Lie structure $(\tr,[\cdot,\cdot])$ induced by $\A\otimes \Der(\A)$ given in Theorem \ref{theo: post-Lie structure from derivations}. To that aim, we fix $\z^{\gamma'}\otimes D^{(\n')},\z^\gamma \otimes D^{(\n)}\in L$, namely $\gamma,\gamma'\in\M$ and $|\n'|<|\gamma'|$, $|\n|<|\gamma|$.

By \eqref{eq:partial12}-\eqref{eq:zgaDntr}, for $\tr$ it remains to prove the two following equalities
\[
(\z^{\gamma'}\otimes D^{(\n')})\tr (\z^\gamma \otimes D^{(\n)})=\z^{\gamma'} D^{(\n')}\z^{\gamma}\otimes D^{(\n)}\in L,
\]
\[
   (\1\otimes \partial_i) \tr (\z^\gamma \otimes D^{(\n)})= \partial_i\z^\gamma\otimes D^{(\n)}\in L.
 \]

Recalling the defining equalities \eqref{eq:D0} and \eqref{eq:partial_i}, we have:
$$\z^{\gamma'} D^{(\n')}\z^{\gamma}\in\left\{ 
    \begin{array}{ll} \text{Span}\{\z^{\gamma'+\gamma+e_{k+1}-e_k}\}_{k\geq 0}\quad &\text{if $\n'=\0$},
\\ \\ \text{Span}\{\z^{\gamma'+\gamma-e_{\n'}}\} &\text{if $\n'\ne\0$},
    \end{array}
    \right.$$
and 
$$\partial_i\z^\gamma\in \text{Span}\{\z^{\gamma+e_{k+1}-e_k+e_{\e_i}}\}_{k\geq 0}\oplus \text{Span}\{\z^{\gamma-e_\n+e_{\n+\e_i}}\}_{\n\in\N^d_*}$$
Using the additivity of the homogeneity, and the fact that $|e_{k+1}-e_k|=0$ and $|e_\n|=|\n|$ (in particular $|e_{\e_i}|=1$ and $|e_{\n+\e_i}|=|\n|+1$), we obtain:
\[
\begin{split}
& |\gamma'+\gamma+e_{k+1}-e_k|=|\gamma'|+|\gamma|+|e_{k+1}-e_k|=|\gamma'|+|\gamma|>|\gamma'|+|\n|>|\n|,
\\ & |\gamma'+\gamma-e_{\n'}|=|\gamma'|+|\gamma|-|\n'|>|\gamma|>|\n|,
\\ & |\gamma+e_{k+1}-e_k+e_{\e_i}|=|\gamma|+|e_{k+1}-e_k|+|e_{\e_i}|=|\gamma|+1>|\gamma|>|\n|,
\\ & |\gamma-e_\n+e_{\n+\e_i}|=|\gamma|-|\n|+|e_{\n+\e_i}|=|\gamma|-|\n| + |\n| +1>|\gamma|>|\n|,
\end{split}
\]
For the bracket $[\cdot,\cdot]$, by \eqref{eq:zgazga}-\eqref{eq:paipaj}-\eqref{eq:zgaDn} what is left is just to prove that
$\z^\gamma\otimes D^{(\n-\e_i)}\in L$ for $\n\ne\0$. Again this is a simple verification based on
$|\n-\e_i|=|\n|-1<|\n|<|\gamma|$.
\end{proof}

We note that our present post-Lie algebra $(L,\tr,[\cdot,\cdot])$ doesn't require any extra condition on the multi-indices $\gamma$ of the elements $\z^\gamma \otimes D^{(\n)}$, unlike the Lie algebra described in \cite[\S 3.10, Lemma 3.3]{LOT}, where an extra grading (denoted $[\gamma]$ there) is needed.

\medskip

\subsection{Two bases for the enveloping algebra}

Recall that in section \ref{subsec: basis env alg} we constructed a basis for the enveloping algebra $\env_{[\cdot,\cdot]}(L_0)$ 
which allows to describe explicitly the product $\trbar$ in a convenient way. It is simple to see that 
$\B_{\env_{[\cdot,\cdot]}(L)}:=\B_{\env_{[\cdot,\cdot]}(L_0)}\cap \env_{[\cdot,\cdot]}(L)$ gives an equally convenient basis for $\env_{[\cdot,\cdot]}(L)$
(recall that $L\subset L_0$ and the two spaces are defined in \eqref{eq: post-Lie algebra L} and \eqref{eq: sub-post-Lie algebra L} respectively).

In particular we obtain that $\B_{\env_{[\cdot,\cdot]}(L)}=\{E_\m F_J\}_{\m,J}$ with
\begin{equation}\label{def:B'}
\begin{split}
E_{\m}&:=\frac{1}{\m!}(1\otimes\partial)^{\m}, 
\qquad \m=(m_1,\ldots,m_d)\in\N^d,
\\
F_{J}&:=\prod_{(\gamma,\n)\in \M\times\N^d} \frac{1}{J(\gamma,\n)!}(\z^{\gamma}\otimes D^{(\n)} )^{J(\gamma,\n)},
\end{split}
\end{equation}
where $J:\M\times\N^d\to\N$ has compact support and satisfies $|\n|<|\gamma|$ for all $(\gamma,\n)$ such that $J(\gamma,\n)>0$.
We use the convention $E_\0=F_{\emptyset}=\mathds{1}\in\env_{[\cdot,\cdot]}(L)$. Recall also the value \eqref{eq: rhobar(E_mF_J)}
of $\rhobar(E_{\m}F_{J})$.

The main technical result in this section is the following Proposition (see \cite[Lemma 4.9]{LOT}), which shows in particular that $L$ satisfies
Assumption \ref{assump: finiteness rhobar} above.
\begin{prop}\label{pr:finite} 
For all $\beta\in\M$ there are only finitely many $u\in\B_{\envU(L)}$ and $\gamma\in\M$ such that $\la \rho(u)(\z^\gamma),\z^\beta\ra\ne 0$.
\end{prop}
\begin{proof}
    By Proposition \ref{prop: finiteness rhobar}, it suffices to prove that Assumption \ref{assump: finiteness rhobar} is satisfied, therefore when $u\in\B_{L}$, we have that
    \begin{itemize}[leftmargin=*]
        \item if $u=\1\otimes \partial_i$, then $\rho(\1\otimes \partial_i)(\z^\gamma)=\partial_i\z^\gamma$ and in that case, by \eqref{eq:partial_i}:
        \begin{multline*}
        \la \partial_i \z^\gamma,\z^\beta\ra\ne 0\Rightarrow\\
        (\exists k\geq 0,~\beta=\gamma+e_{k+1}-e_k+e_{\e_i})~ \vee ~ (\exists\n\in\N^d_*,~\beta=\gamma-e_\n+e_{\n+\e_i}).
        \end{multline*}
        In either case, by \eqref{eq:|}
        \[|\beta| =|\gamma|+1.\]
        \item if $u=\z^{\gamma'}\otimes D^{(\n)}$, then $\rho(\z^{\gamma'}\otimes D)(\z^\gamma)=\z^{\gamma'} D^{(\n)}\z^\gamma$ and by \eqref{eq: tilt}:
$$
\la \z^{\gamma'} D^{(\n)}\z^\gamma,\z^\beta\ra\ne 0~\Rightarrow~ \left\{
    \begin{array}{ll} \text{if $\n=\0$}:\quad
    \exists k\geq 0,~ \beta=\gamma'+\gamma+e_{k+1}-e_k,
\\ \\  \text{if $\n\ne\0$}:\quad  \beta=\gamma'+\gamma-e_\n,
    \end{array}
    \right.
$$
    \end{itemize}
and in both cases by \eqref{eq:|} 
\[|\beta|=|\gamma| + |\gamma'|-|\n|.\]

Since by definition \eqref{eq: sub-post-Lie algebra L} of the space $L$: $\z^{\gamma'}\otimes D^{(n)}\in L\Rightarrow|\gamma'|>|\n|$, we obtain in all cases that $|\gamma|<|\beta|$. Then since, $0<{\gamma'}$, we have from the definition of the homogeneity \eqref{eq:|} that there are only finitely many possible $\gamma$'s. For each $\n\ne\0$, we have that $D^{(\n)}\z^\gamma=0$ unless $\gamma_{\n}>0$; since we have already selected finitely many possible $\gamma$'s, each with compact support in $\N$, there are only finitely many such $\n$'s. Then, for a choice of such $\gamma$ and $\n$, again from the definition of the homogeneity \eqref{eq:|}, there are finitely many $\gamma'\in\M$ such that $|\gamma'|= |\beta| - |\gamma|+|\n|$.
\end{proof}

We now introduce the basis ${\overline\B}_{\envU(L)}=\{\overline E_\m \overline F_J\}_{\m,J}$, corresponding to \eqref{eq:dualityL}:
\begin{equation*}
\begin{split}
\overline E_{\m}&:=(1\otimes\partial)^{\m}, 
\qquad \m=(m_1,\ldots,m_d)\in\N^d,
\\
\overline F_{J}&:=\prod_{(\gamma,\n)\in \M\times\N^d} (\z^{\gamma}\otimes D^{(\n)} )^{J(\gamma,\n)},
\end{split}
\end{equation*}
where $J:\M\times\N^d\to\N$ has compact support and satisfies $|\n|<|\gamma|$ for all $(\gamma,\n)$ such that $J(\gamma,\n)>0$. 

Note that $\overline E_\m \overline F_J=(\m! J!) \, E_\m F_J$, or in other words
$T:\B_{\envU(L)}
\to{\overline\B}_{\envU(L)}$ as in \eqref{eq:T} is given by
\begin{equation}\label{eq:T2}
T(E_\m F_J) = \overline E_\m \overline F_J.
\end{equation}

The two bases $\B_{\envU(L)}$ and ${\overline\B}_{\envU(L)}$ are in duality via \eqref{eq:T}, namely
\begin{equation}\label{eq:T'}
\la E_\m F_J,\overline E_{\bar\m}\overline F_{\bar J} \ra = \ind_{(\m=\bar\m, \, J=\bar J)}.
\end{equation}
The multiplication table of the $\ast$-product \eqref{coro: commutative product env alg} in $\env_{[\cdot,\cdot]}(L)$, 
in duality with the coproduct \eqref{eq:coco} with respect to the pairing \eqref{eq:T'}, is (see \cite[(4.43)]{LOT})
\begin{equation*}
    (\overline E_{\m}\overline F_J)\ast (\overline E_{\bar\m}\overline F_{\bar J})=\overline E_{\m+\bar\m}\, \overline F_{J+ \bar J}.
\end{equation*}

\medskip

\subsection{The space of formal series}
Set now $\overline\A:=\R[[\z_k,\z_\n]]_{k\in\N,\n\in\N^d_*}$, the space of formal series in the commuting variables
$\{\z_k,\z_\n\}_{k\in\N,\n\in\N^d_*}$. Then $a\in\overline\A$ can be written
\[
a=\sum_{\gamma\in\M}a_\gamma\z^\gamma,
\]
and $\overline\A$ turns out to be a commutative algebra with product
\[
ab = \sum_{\gamma\in\M}\left[\sum_{\gamma_1+\gamma_2=\gamma} a_{\gamma_1}b_{\gamma_2}\right]\z^\gamma.
\]

We have a canonical pairing between $\overline\A$ and $\A$, which is the bilinear extension of
\begin{equation}\label{eq:duality}
\la \sum_{\gamma\in\M} a_\gamma \, \z^\gamma,\z^\beta\ra = a_\beta, \qquad \beta\in\M.
\end{equation}

In this way we have a canonical identification between $\overline\A$ and the dual $\A^*$ of $\A$.
Then Proposition \ref{pr:finite} has the following important consequence.
\begin{prop}\label{pr:well-def}
For all $f:\envU(L)\to\R$ linear, the following map is well-defined and linear: $\rhobar(f):\overline\A\to\overline\A$,
\[
\rhobar(f)\left(\sum_{\gamma\in\M} a_\gamma \, \z^\gamma\right):=\sum_{\beta\in\M} 
\left[\sum_{\gamma\in\M}\sum_{u\in\B_{\envU(L)}}a_\gamma \, f\left(Tu\right) \, 
\la\rhobar(u)(\z^\gamma),\z^\beta\ra \right] \z^\beta.\]
\end{prop}
\begin{proof}
This is a consequence of Proposition \ref{pr:finite} since for $\z^\beta$ fixed the set of $(\gamma,u)\in\M\times\B_{\envU(L)}$ such
that $\la\rhobar(u)(\z^\gamma),\z^\beta\ra\ne 0$ is finite.
\end{proof}
\medskip

\subsection{Characters}\label{sec:char}

Now we add a crucial multiplicativity hypothesis on $f:\envU(L)\to\R$ for the commutative product $\ast$ defined in \eqref{eq:*}-\eqref{coro: commutative product env alg}. We suppose that $f$ is a character on $(\envU(L),\ast)$, namely
\[
f\left((1\otimes\partial)^{\m}
\prod_{i=1}^k \z^{\gamma_i}\otimes D^{(\n_i)} \right) = \prod_{i=1}^d (f(\1\otimes\partial_i))^{m_i}\prod_{i=1}^k f\left(\z^{\gamma_i}\otimes D^{(\n_i)}\right).
\]
This leads to the following key proposition (see \cite[Proposition 5.1-(ii)]{LOT})
\begin{prop}\label{prop: multiplicativity Gamma_f^*}
If $f$ is a character of the commutative algebra $(\env_{[\cdot,\cdot]}(L),\ast)$, then the map $\rhobar(f):\overline\A\to\overline\A$ is an algebra morphism, namely it verifies the following multiplicativity property, for all $a,b\in\overline\A$:
\begin{equation*}
   \rhobar(f)(ab)=\rhobar(f)(a)\ \rhobar(f)(b).
\end{equation*}
\end{prop}
\begin{proof}
By Proposition \ref{pr:well-def}
\[
\rhobar(f)(ab) = \sum_{\beta\in\M}\left[\sum_{u\in\B_{\envU(L)}} f\left(Tu\right)\sum_{\gamma_1,\gamma_2\in\M} a_{\gamma_1}b_{\gamma_2}\,
\la\rhobar(u)(\z^{\gamma_1}\z^{\gamma_2}),\z^\beta\ra\right]\z^\beta.
\]
By Proposition \ref{prop: multiplicativity property representation} and \eqref{eq:copast}, 
for $a,b\in\overline\A$ we have
\[
\rhobar\left(u\right)(\z^{\gamma_1}\z^{\gamma_2}) =\sum_{u_1,u_2\in\B_{\envU(L)}} \ind_{\big(Tu=(Tu_1)*(Tu_2)\big)} \,
\rhobar\left(u_1\right)(\z^{\gamma_1})\ \rhobar\left(u_2\right)(\z^{\gamma_2}).
\]
Now $\la ab,\z^\beta\ra = \sum_{\beta_1+\beta_2=\beta}\la a,\z^{\beta_1}\ra\la b,\z^{\beta_2}\ra$, so that 
\[
\begin{split}
&\la\rhobar(u_1)(\z^{\gamma_1})\ \rhobar(u_2)(\z^{\gamma_2}),\z^\beta\ra = \sum_{\beta_1+\beta_2=\beta}
\la\rhobar(u_1)(\z^{\gamma_1}),\z^{\beta_1}\ra\la\rhobar(u_2)(\z^{\gamma_2}),\z^{\beta_2}\ra.
\end{split}
\]
By the character property, $f\left(Tu\right)=f\left(Tu_1\right)f\left(Tu_2\right)$, and this allows to conclude the proof.
\end{proof}
In particular, if $f$ is a character on $(\env_{[\cdot,\cdot]}(L),\ast)$ then 
\[
\rhobar(f)(\z^\gamma)=\prod_{i\in\N\sqcup\N^d_*} (\rhobar(f)(\z_i))^{\gamma_i},
\]
and for $a\in\overline\A$
\[
\begin{split}
\rhobar(f)(a)&=\sum_{\beta\in\M} 
\left[\sum_{\gamma\in\M}a_\gamma \,  \la\rhobar(f)(\z^\gamma),\z^\beta\ra \right] \z^\beta
\\ & = \sum_{\beta\in\M} 
\left[\sum_{\gamma\in\M}a_\gamma \, \la\prod_{i\in\N\sqcup\N^d_*} (\rhobar(f)(\z_i))^{\gamma_i},\z^\beta\ra \right] \z^\beta.
\end{split}
\]
In other words we have proved the following.
\begin{lemma}\label{lem:determ}
If $f$ is a character on $(\env_{[\cdot,\cdot]}(L),\ast)$ then for any $a\in\overline\A$ the value of $\rhobar(f)(a)$ is uniquely determined by
the values of $(\rhobar(f)(\z_i))_{i\in\N\sqcup\N^d_*}$.
\end{lemma}

By Lemma \ref{lem:determ}, it is very important to compute the value of the representation $\rhobar$ on the elements $\{\z_k,\z_\n\}_{k\in\N,\n\in\N^d_*}$.
This will be done in Section \ref{sec:expl} below. We first give a preparatory lemma.

\begin{lemma}\label{lem: composition of derivations}
For all $\ell,k\in\N$, $\m,\n_1,\ldots,\n_\ell\in\N^d$ and $\n\in\N^d_*$ we have
\begin{align*}
        \partial^\m\circ D^{(\n_1)}\circ \cdots \circ D^{(\n_\ell)}(\z_\n)&=\left\{
    \begin{array}{ll}
    \partial^\m\z_\n \quad &\text{if}~ \ell=0,\\
    \1 &\text{if}~\ell=1,~\n_1=\n,~\m=\0,\\
    0 &\text{otherwise},
    \end{array}
    \right.\\
    \partial^\m\circ D^{(\n_1)}\circ \cdots \circ D^{(\n_\ell)}(\z_k)&=\left\{
    \begin{array}{ll}
    \partial^\m(D^{(\0)})^{\circ\ell}\z_k~&\text{if}~\ell=0,
    \\ &\text{or}~\n_1=\ldots=\n_\ell=\0,\\
    0 &\text{otherwise}.
    \end{array}
    \right.
\end{align*}
 The following equalities are verified for all $\n,\m\in \N^d, \n\neq\0$ and $k,\ell\in\N$:
\begin{align}
\label{eq: (D^{(0)})^l z_k}
(D^{(\0)})^{\circ\ell}\z_k&=\frac{(k+\ell)!}{k!}\z_{k+\ell},\\
    \label{eq: 1/m! partial^m z_n}
       \frac1{\m!}\partial^\m\z_\n &=  \binom{\n+\m}{\n} \z_{\n+\m},
\\
\label{eq: 1/m! partial^m z_k}
    \frac1{\m!}\partial^\m\z_k &= \sum_{\ell\geq 0} \binom{k+\ell}{k} \z_{k+\ell} \sum_{\substack{\m_1,\ldots,\m_\ell\in\N^d_*\\\m_1+\cdots+\m_\ell=\m}}\z_{\m_1}
\cdots\z_{\m_\ell}.
\end{align}
\end{lemma}
\begin{proof} The first equality is obtained since $D^{(\n')}(\z_\n)=\delta_{\n',\n}$, for all $\n,\n'\in\N^d$, $\n\neq \0$. The second one is obtained since $D^{(\n)}\z_k=0$ for all $\n\neq \0$.

Now, \eqref{eq: (D^{(0)})^l z_k} follows easily from the definition of $D^{(\0)}$. Let us recall that:
\[\partial_i=\sum_{\n\in\N^d} (n_i+1)\z_{\n+\e_i} D^{(\n)}.\]
Thus for $\n\in\N^d_*$ and $k\in\N$, we have
\[
\partial_i \z_\n = (n_i+1)\z_{\n+\e_i},
\qquad
\partial_i \z_k = \z_{\e_i} (k+1)\z_{k+1},
\]
so that in particular \eqref{eq: 1/m! partial^m z_n} follows easily by recurrence on $\m\in\N^d$.

We prove now \eqref{eq: 1/m! partial^m z_k} by recurrence on $\m\in\N^d$. The base case $\m=\0$
is trivial since the right-hand side reduces to the case $\ell=0$; we suppose now that the formula is proved for $\m=(m_1,\ldots,m_d)\in\N^d$ and we show it (for example) for $\m+\e_1=(m_1+1,m_2,\ldots,m_d)$. First we have
\begin{multline*}
\partial_1 (\z_{k+\ell}\z_{\m_1}\cdots\z_{\m_\ell})= 
 (k+\ell+1)\z_{k+\ell+1}\z_{\e_1}\z_{\m_1}\cdots\z_{\m_\ell}\\
 +\z_{k+\ell}\sum_{i=1}^\ell \z_{\m_1}\cdots(m^i_1+1)\z_{\m_i+\e_1}\cdots\z_{\m_\ell},
\end{multline*}
where we recall that $\e_1=(1,0,\ldots,0)\in\N^d$ and we note $\m_i=(m^i_1,\ldots, m^i_d)\in\N^d$. Now
\[
\begin{split}
&\sum_{\ell\geq 0} \binom{k+\ell}{k} (k+\ell+1)\z_{k+\ell+1}\,\z_{\e_1}\sum_{\substack{\m_1,\ldots,\m_\ell\in\N^d_*\\
\m_1+\cdots+\m_\ell=\m}}\z_{\m_1}\cdots\z_{\m_\ell}
\\ & = \sum_{\ell\geq 0} \binom{k+\ell}{k} \z_{k+\ell}\sum_{\substack{\m_1,\ldots,\m_\ell\in\N^d_*\\\m_1+\cdots+\m_\ell=\m+\e_1}}\z_{\m_1}\cdots\z_{\m_\ell}\sum_{i=1}^\ell \ind_{(\m_i=\e_1)}.
\end{split}
\]
On the other hand we have
\[
\begin{split}
&\sum_{\ell\geq 0} \binom{k+\ell}{k} \z_{k+\ell}\sum_{i=1}^\ell \sum_{\substack{\m_1,\ldots,\m_\ell\in\N^d_*\\\m_1+\cdots+\m_\ell=\m}}\z_{\m_1}\cdots(m^i_1+1)\z_{\m_i+\e_1}\cdots\z_{\m_\ell}
\\ & = \sum_{\ell\geq 0} \binom{k+\ell}{k} \z_{k+\ell} \sum_{\substack{\m_1,\ldots,\m_\ell\in\N^d_*\\\m_1+\cdots+\m_\ell=\m+\e_1}}\z_{\m_1}\cdots\z_{\m_\ell} \sum_{i=1}^\ell m^i_1\ind_{(\m_i\ne \e_1)}.
\end{split}
\]
Therefore
\[
\begin{split}
 \partial_1 \frac1{\m!}\partial^\m\z_k &= \sum_{\ell\geq 0} \binom{k+\ell}{k}\sum_{\substack{\m_1,\ldots,\m_\ell\in\N^d_*\\\m_1+\cdots+\m_\ell=\m}}\partial_1\left(  \z_{k+\ell}\z_{\m_1}\cdots\z_{\m_\ell}\right)
\\ & = \sum_{\ell\geq 0} \binom{k+\ell}{k} \z_{k+\ell} \sum_{\substack{\m_1,\ldots,\m_\ell\in\N^d_*\\\m_1+\cdots+\m_\ell=\m+\e_1}}\z_{\m_1}\cdots\z_{\m_\ell} \sum_{i=1}^\ell m^i_1
\\ & = (m_1+1)\sum_{\ell\geq 0} \binom{k+\ell}{k} \z_{k+\ell} \sum_{\substack{\m_1,\ldots,\m_\ell\in\N^d_*\\\m_1+\cdots+\m_\ell=\m+\e_1}}\z_{\m_1}\cdots\z_{\m_\ell}
\end{split}
\]
and therefore \eqref{eq: 1/m! partial^m z_k} is proved.
\end{proof}
Formula \eqref{eq: 1/m! partial^m z_k} is \cite[formula (A.5)]{LOT}, where it is proved as an application of the Faà di Bruno identity.
\medskip

\subsection{Explicit formulae}\label{sec:expl}

Let us consider the space $L$ previously defined by \eqref{eq: sub-post-Lie algebra L} along with its basis $\B_{\envU(L)}$ and a character
$f$ on $(\env_{[\cdot,\cdot]}(L),\ast)$.
Then $f$ is entirely characterised by its values on the basis elements of $L$:
\begin{itemize}
    \item $f(\1\otimes \partial_i)$, for all $i\in\{1,\ldots,d\}$;
    \item $f(\z^{\gamma}\otimes D^{(\n)})$, for all $\gamma\in\M,~\n\in\N^d,~|\gamma|>|\n|$.
\end{itemize}
We use the notation
\[
f(\1\otimes\partial)^\m:=f(\1\otimes\partial_1)^{m_1}\cdots f(\1\otimes\partial_d)^{m_d}, \qquad \m=(m_1,\ldots,m_d).
\]

\medskip
Following Proposition \ref{prop: multiplicativity Gamma_f^*}, the map $\rhobar(f):\overline\A\to\overline\A$ is entirely determined by its values on basis elements $\{\z^\gamma\}_{\gamma\in\M}$ of $\A$. Applying formula \eqref{def:B'} to our present setting, one has (see \eqref{def:B'}):
\begin{align}\nonumber
\rhobar(f)(\z^\gamma)=&\sum_{\m,J} f(\1\otimes\partial)^\m
\prod_{(\beta,\n)} (f(\z^{\beta}\otimes D^{(\n)}) )^{J(\beta,\n)}
\\&  \cdot\left[\prod_{(\beta,\n)} \frac{1}{J(\beta,\n)!}(\z^{\beta} )^{J(\beta,\n)}\right] \frac1{\m!}\partial^\m\circ 
\left[\prod_{(\beta,\n)} (D^{(\n)})^{\circ J(\beta,\n)}\right](\z^\gamma).
\label{eq: explicit formula Gamma_f^*(z^gamma)}
\end{align}

\begin{notations}
    We set $f^{(\n)}\in\overline\A$ for all $\n\in\N^d$:
\begin{equation*}
    f^{(\n)}:=\sum_{\m\in
\N^d_*} \binom{\n+\m}{\n} f(\1\otimes\partial)^\m \,\z_{\n+\m} + \sum_{\substack{\beta\in\M\\ |\beta|>|\n|}}f(\z^{\beta}\otimes D^{(\n)})\,\z^{\beta}.
\end{equation*}

In particular:
\begin{equation}\label{eq: f^{(0)}}
    f^{(\0)}:=\sum_{\m\in
\N^d_*} f(\1\otimes\partial)^\m \z_{\m} + \sum_{\beta\in\M} f(\z^{\beta}\otimes D^{(\0)})\,\z^{\beta}.
\end{equation}
\end{notations}

\medskip
Then we have (see \cite[(5.17)-(5.18)]{LOT}) the following proposition.
\begin{prop} The map $\rhobar(f):\overline\A\to\overline\A$ satisfies for $\n\in\N^d_*$ and $k\in\N$
\begin{align}
    \rhobar(f)(\z_\n)&=\z_\n+f^{(\n)}\label{eq: Gamma_x^*(z_n)},\\
    \rhobar(f)(\z_k)&=\sum_{\ell\geq 0} \binom{k+\ell}{k} \left( f^{(\0)} \right)^\ell \z_{k+\ell}\label{eq: Gamma_x^*(z_k)}.
\end{align}
\end{prop}

\begin{proof}
    The two equalities are obtained with formula \eqref{eq: explicit formula Gamma_f^*(z^gamma)} and Lemma \ref{lem: composition of derivations}.
    The first equality \eqref{eq: Gamma_x^*(z_n)} is straightforward. 
    For the second equality \eqref{eq: Gamma_x^*(z_k)}, on one side for $\ell\in \N$ fixed, we have
\[
\left( f^{(\0)} \right)^\ell= \sum_{p+q=\ell} \frac{\ell!}{p!q!} \left[\sum_{\beta\in\M} f(\z^{\beta}\otimes D^{(\0)})\,\z^{\beta}\right]^p
\left[\sum_{\m\in
\N^d_*} f(\1\otimes\partial)^\m \z_{\m}\right]^q;
\]

\newpage
and by the multinomial theorem

\begin{align*}
&\left[\sum_{\beta\in\M} f(\z^{\beta}\otimes D^{(\0)})\,\z^{\beta}\right]^p\\
&= \sum_{k:\M\to\N} \ind_{\left(\sum_{\beta}k_\beta=p\right)} p!\prod_{\beta\in\M}\left[\frac{1}{k_\beta!}
\left(f(\z^{\beta}\otimes D^{(\0)}) \,\z^{\beta}\right)^{k_\beta}\right],
\end{align*}

while
\[
\begin{split}
\left[\sum_{\m\in\N^d_*} f(\1\otimes\partial)^\m \z_{\m}\right]^{q}&=
\sum_{\m_1,\ldots,\m_{q}\in\N^d_*} 
  \prod_{i=1}^{q}\left[f(\1\otimes\partial)^{\m_i}\z_{\m_i}\right]
\\ & = \sum_{\m_1,\ldots,\m_{q}\in\N^d_*}f(\1\otimes\partial)^{\m_1+\cdots+\m_{q}}\z_{\m_1}\cdots
\z_{\m_{q}}
\\ & = \sum_{\m\in\N^d_*}f(\1\otimes\partial)^{\m}\sum_{\m_1+\cdots+\m_{q}=\m}\z_{\m_1}\cdots
\z_{\m_{q}}.
\end{split}
\] 
By \eqref{eq: 1/m! partial^m z_k}, denoting
\[
V(p):=\sum_{k:\M\to\N}  \ind_{\left(\sum_{\beta}k_\beta=p\right)}
  \prod_{\beta\in \M}\left[\frac{1}{k_\beta!}
\left(f(\z^{\beta}\otimes D^{(\0)}) \,\z^{\beta}\right)^{k_\beta}\right],
\]
we obtain that
\[
\begin{split}
\sum_{\ell\geq 0} \binom{k+\ell}{k} &\left( f^{(\0)} \right)^\ell \z_{k+\ell}
\\ & = \sum_{p\geq 0}\sum_{q\geq 0} \frac{(k+p+q)!}{k!\,b!} \,\z_{k+p+q} \, V(p)\,
\left[\sum_{\m\in\N^d_*} f(\1\otimes\partial)^\m \z_{\m}\right]^{q}
\\ & = \sum_{\m\in\N^d_*}\frac{f(\1\otimes\partial)^{\m}}{\m!} V(p) \,
\partial^\m\sum_{p\geq 0} \frac{(k+p)!}{k!}  \z_{k+p}.
\end{split}
\]
By \eqref{eq: (D^{(0)})^l z_k}-\eqref{eq: 1/m! partial^m z_k} we obtain since $D^{(\n)}\z_k=0$ for any $\n\ne\0$

\[
\begin{split}
\sum_{\ell\geq 0} \binom{k+\ell}{k} \left( f^{(\0)} \right)^\ell \z_{k+\ell}& = \sum_{\substack{\m\in\N^d_*\\p\geq 0}}\frac{f(\1\otimes\partial)^{\m}}{\m!}V(p)\,\partial^\m (D^{(\0)})^{\circ p}\, \z_k
\\ & = \rhobar(f)(\z_k),
\end{split}
\]
which is the desired equality.
\end{proof}
\medskip

\subsection{Graded Hopf algebra and its graded dual}\label{sec:group}
We define a homogeneity $|\cdot|:\B_{\envU(L)}\to\R_+$ by
\[
|\z^\beta\otimes D^{(\n)}|:=|\beta|-|\n|, \quad 
|\1\otimes\partial_i|:=1, \quad |\ind|:=0, \quad |u_1u_2|:=|u_1|+|u_2|.
\]

We set $A:=\alpha\N+\N=\{\alpha i+j: i,j\in\N\}$. By \eqref{eq:|} the homogeneity $|\beta|$ of $\beta\in\M$
takes values in $A$.
This allows to grade $\envU(L)$ setting 
\[
U_\kappa:=\R\{u\in\B_{\envU(L)}\}_{|u|=\kappa}, \quad \kappa\in A, 
\]
so that $\envU(L)=\oplus_{\kappa\in A} U_\kappa$. It is easy to check from the definitions that 
this makes $(\envU(L),\trbar,\Cop)$ a graded and connected (namely $U_0=\R\{\1\}$) bialgebra. This gives a more direct proof of the existence of an antipode for 
$(\envU(L),\trbar,\Cop)$, with respect to the general setting used in \cite{ebrahimi2014lie,MHandbook}.

By Proposition \ref{pr:finite}, Assumption \ref{assump: finiteness rhobar}, is satisfied in this setting, which implies by Lemma \ref{lem: implication assumptions} that Assumption \ref{assump: finiteness tr} is also satisfied. Moreover, by Remark \ref{rem: finiteness Lie bracket AtensDer(A)}, Assumption \ref{assump: finiteness Lie bracket}, or equivalently Assumption \ref{assump: finiteness Lie bracket AtensDer(A)}, is also satisfied.
We can therefore define a dual bialgebra structure $(\envU(L),\ast,\Delta_\trbar)$ as in \eqref{eq: Delta_trbar} and in Proposition \ref{prop:trbar}, where
$\Delta_\trbar:\envU(L)\to\envU(L)\otimes\envU(L)$ is defined with respect to the pairing \eqref{eq:T'} by
\[
\Delta_\trbar u:=\sum_{u_1,u_2\in\B_{\envU(L)}} \la u_1\trbar u_2,u \ra \, Tu_1 \otimes Tu_2,
\]
and $T:\B_{\envU(L)}\to\overline\B_{\envU(L)}$ is given by \eqref{eq:T2}.
Moreover $(\envU(L),\ast,\Delta_\trbar)$ is graded by the homogeneity as well and it is also connected (which confirms that it is indeed a Hopf algebra).

Then the set 
$$H:=\{f\in\envU(L)^*: f(\mathds{1})=1\}$$
forms a group for the product for $f_1,f_2\in H$
\[
f_1\trbar  f_2 (u) = \la f_1\otimes f_2,\Delta_\trbar u\ra, \qquad u\in\envU(L),
\]
and the set $G$ of real-valued characters on $(\envU(L),\ast)$, defined as the set of all $f:\envU(L)\to\R$ such that $f(\mathds{1})=1$ and
\[
f(u_1\ast u_2) = f(u_1)\, f(u_2), \qquad u_1,u_2\in\envU(L),
\]
is a subgroup of $H$,
see Definition \ref{def:grch}, Proposition \ref{pr:grch} and Section \ref{sec:char}.
Then Proposition \ref{pr:well-def} tells us that we have a well-defined extension of $\rhobar:G\to {\rm End}(\overline\A)$. Moreover
by Proposition \ref{prop: rhobar group morphism} the map $f\mapsto \rhobar(f)$ is a group morphism from $(G,\trbar,\mathds1)$ to $({\rm Aut}(\overline\A),\circ,\mathrm{id})$.\\

Finally, we note that in \cite{LOT} the relevant module (or comodule) is the one constructed in
Section \ref{sec:modcomod} above, while in the first constructions of regularity structures \cite{Hai14,BHZ}
the definition is slightly different. We show now how to obtain the object used in \cite{Hai14,BHZ},
based on the one use in \cite{LOT} and Section \ref{sec:modcomod} above.\\

We define now the linear map $\Lambda:\envU(L)^*\otimes\overline\A\to\overline\A$ by
\begin{equation*}
\begin{split}
&\Lambda(f\otimes a):=\rhobar(f)(a)+\la a,\1\ra f^{(\0)}
\end{split}
\end{equation*}
with
\[
    f^{(\0)}:=\sum_{\m\in
\N^d_*} f\left((\1\otimes\partial)^\m\right) \z_{\m} + \sum_{\beta\in\M} f\left(\z^{\beta}\otimes D^{(\0)}\right)\z^{\beta}\in\overline\A
\]
in the notation \eqref{eq: f^{(0)}} (which however was introduced only for $f$ a character, while here $f$ is a generic element
of $\envU(L)^*$).

\begin{prop}\label{prop:module} We have that 
$\left(\overline\A,\Lambda\right)$ is a left $(\envU(L)^*,\trbar)$-module, namely for all $f_1,f_2\in\envU(L)^*$ and $a\in\overline\A$, we have
\[
\Lambda\left( (f_1\trbar  f_2)\otimes a\right)=\Lambda\left( f_1\otimes \Lambda(f_2\otimes a)\right).
\]
\end{prop}
\begin{proof}
We have
\[
\begin{split}
&\Lambda\left( (f_1\trbar  f_2)\otimes a\right) = \rhobar(f_1)(\rhobar(f_2)(a))+\la a,\1\ra (f_1\trbar  f_2)^{(\0)}
\\ & \Lambda\left( f_1\otimes \Lambda(f_2\otimes a)\right)=\rhobar(f_1)\left(\rhobar(f_2)(a)+\la a,\1\ra f_2^{(\0)}\right)+
\la \rhobar(f_2)(a),\1\ra  f_1^{(\0)}
\end{split}
\]
where in the second equality we have used that $\la f_2^{(\0)},\1\ra=0$.
We want now to prove that 
\[
\la \cdot,\1\ra(f_1\trbar  f_2)^{(\0)}=\la \rhobar(f_2)(\cdot),\1\ra f_1^{(\0)}+
\la \cdot,\1\ra \rhobar(f_1)(f_2^{(\0)}),
\]
namely
\[
(f_1\trbar  f_2)^{(\0)}=f_2(\ind) f_1^{(\0)}+\rhobar(f_1)(f_2^{(\0)}),
\]
since $\la \rhobar(f_2)(\z^\gamma),\1\ra=0$ for any $\gamma\ne0$ while $\la \rhobar(f_2)(\1),\1\ra=f_2(\ind)$.
%
We set
\[
\z(u):=
\sum_{\m\in\N^d_*}\la u,(\1\otimes\partial)^\m\ra\z_\m+
\sum_{\beta\in\M} \la u,\z^\beta\otimes D^{(\0)}\ra\z^\beta \in\overline\A
\]
for $u\in\envU(L)$, so that
\[
f^{(\0)} = \sum_{u\in\B_{\envU(L)}} f(Tu) \,\z(u).
\]
Then
\[
(f_1\trbar  f_2)^{(\0)}=\sum_{u_1,u_2\in\B_{\envU(L)}} f_1(Tu_1)f_2(Tu_2)\, \z(u_1\trbar u_2),
\]
while
\[
\begin{split}
&f_2(\ind) f_1^{(\0)}+\rhobar(f_1)(f_2^{(\0)})
\\ &=\sum_{u_1,u_2\in\B_{\envU(L)}} f_1(Tu_1) f_2(Tu_2) \left[
\ind_{(u_2=\ind)} \z(u_1)+\rhobar(u_1)(\z(u_2))\right].
\end{split}
\]
Therefore all we have to prove is the formula
\[
\z(u_1\trbar u_2)=\ind_{(u_2=\ind)} \z(u_1)+\rhobar(u_1)(\z(u_2)), \qquad \forall u_1,u_2\in\B_{\envU(L)}.
\]
If $u_2=\ind$ then 
this reduces to $\z(u_1)=\z(u_1)$, since $\z(\ind)=0$. If $u_2\ne\ind$ we have to show that
\[
\z(u_1\trbar  u_2)=\rhobar(u_1)(\z(u_2)), \qquad \forall u_1\in\B_{\envU(L)}.
\]
For $u_1=\ind$ this formula reduces to $\z(u_2)=\z(u_2)$. 

We consider therefore $u_1,u_2\in\B_{\envU(L)}\setminus\{\ind\}$. By the definition of $\star$ given by formula \eqref{eq: post-Lie associative product} and using \eqref{eq:trbargeneral2} (or equivalently \eqref{eq:deltatrbar}) we compute for $\beta\in\M$:
\[
\begin{split}
\la u_1\trbar  u_2,\z^\beta\otimes D^{(\0)}\ra&=\la u_1\tr u_2,\z^\beta\otimes D^{(\0)}\ra \\
&=\sum_{\gamma\in\M}\la \z^\beta,\rhobar(u_1)(\z^\gamma)\ra \la u_2,\z^\gamma\otimes D^{(\0)}\ra,
\end{split}
\]
and using moreover \eqref{eq:partial12} and \eqref{eq:zgaDntr}, we compute for $\m\in\N^d_*$:
\[
\begin{split}
\la u_1\trbar  u_2,(\1\otimes\partial)^\m\ra&= \la u_1 u_2,(\1\otimes\partial)^\m\ra\\
&=\sum_{\0\le\n\le\m}\binom{\m}{\n} \la u_1,(\1\otimes\partial)^{\m-\n}\ra \la u_2,(\1\otimes\partial)^\n\ra.
\end{split}
\]
On the other hand
\[
\begin{split}
\rhobar(u_1)(\z(u_2))=&\sum_{\m\in\N^d_*}\la u_2,(\1\otimes\partial)^\m\ra\rhobar(u_1)(\z_\m)
\\ & + \sum_{\beta\in\M} \la u_2,\z^\beta\otimes D^{(\0)}\ra\rhobar(u_1)(\z^\beta).
\end{split}
\]
This shows that, unless $u_2\in\{\frac1{\q!}(\1\otimes\partial)^\q,\z^\gamma\otimes D^{(\0)}:\q\in\N^d_*,\gamma\in\M\}$, we have $\z(u_1\trbar  u_2)=\rhobar(u_1)(\z(u_2))=0$.
\begin{itemize}[leftmargin=*]
    \item If $u_2=\z^\gamma\otimes D^{(\0)}$, since $\la u_1\star (\z^\gamma\otimes D^{(\0)}),(\1\otimes\partial)^\m\ra=0$, for all $\m\in\N^d_\ast$, and all $u_1\in\B_{\envU(L)}$, the desired formula follows from
\begin{align*}
    \z(u_1\trbar  u_2)&=\sum_{\beta\in\M} \la u_1\star (\z^\gamma\otimes D^{(\0)}),\z^\beta\otimes D^{(\0)}\ra\z^\beta\\
    &=\sum_{\beta\in\M} \la u_1\tr (\z^\gamma\otimes D^{(\0)}),\z^\beta\otimes D^{(\0)}\ra\z^\beta\\
    &=\sum_{\beta\in\M}\la \z^\beta,\rhobar(u_1)(\z^\gamma)\ra \z^\beta\\
    &=\rhobar(u_1)(\z^\gamma)\\
    &=\rhobar(u_1)(\z(u_2)).
\end{align*}

\item If $u_2=\frac1{\q!}(\1\otimes\partial)^\q$, since $u_1\tr (\1\otimes\partial)^\q=0$, for all $u_1\in\B_{\envU(L)}$, we have that $u_1\star (\1\otimes\partial)^\q=u_1 (\1\otimes\partial)^\q$, and we obtain
\[
\begin{split}
\z(u_1\trbar  u_2) =& \frac1{\q!}\sum_{\m\in\N^d_*}\la u_1 (\1\otimes\partial)^\q,(\1\otimes\partial)^\m\ra\z_\m\\ 
& + \sum_{\beta\in\M} \la u_1 (\1\otimes\partial)^\q,\z^\beta\otimes D^{(\0)}\ra\z^\beta.
    \end{split}
\]
On one hand by formula \eqref{eq: inversion order basis simpler form}:
\begin{align*}
    &\frac1{\q!}\sum_{\m\in\N^d_*}\la u_1 (\1\otimes\partial)^\q,(\1\otimes\partial)^\m\ra\z_\m\\
    &=\left\{
    \begin{array}{ll}
    0\qquad &\text{if}\quad u_1\notin\{\frac1{\n!}(\1\otimes\partial)^\n:\n\in\N^d_*\},
\\ \\    \sum_{\m\ge\q}\frac{\m!}{(\m-\q)!} \la u_1,(\1\otimes\partial)^{\m-\q}\ra \z_\m\\\qquad=(\n+\q)!\,\z_{\n+\q}\quad &\text{if}\quad u_1=\frac1{\n!}(\1\otimes\partial)^\n.
    \end{array}
    \right.
\end{align*}
On the other hand, again by formula \eqref{eq: inversion order basis simpler form}:
\begin{align*}
    &\la u_1 (\1\otimes\partial)^\q,\z^\beta\otimes D^{(\0)}\ra=\left\{
    \begin{array}{ll}
    0 &\text{if}\quad u_1\notin\{\z^\beta \otimes D^{(\n)}:~\beta\in\M,~\n\in\N^d_*\},
\\ \\  1   &\text{if}\quad u_1=\z^\beta \otimes D^{(\q)}.
    \end{array}
    \right.
\end{align*}
Now for $u_2=\frac1{\q!}(\1\otimes\partial)^\q$ we have
$\z(u_2)=\q!\,\z_\q$ and,
by Lemma \ref{lem: composition of derivations} and equality \eqref{eq: 1/m! partial^m z_k},
\begin{align*}
    \rhobar&(u_1)(\z(u_2))\\
    &=\q!\,\rhobar(u_1)(\z_\q)\\
    &=\left\{
    \begin{array}{ll}
    0\quad &\text{if}\quad u_1\notin\left\{\frac1{\n!}(\1\otimes\partial)^\n,~\z^\beta \otimes D^{(\n)}:\n\in\N^d_*,~\beta\in\M\right\},\\\\
    (\n+\q)!\,\z_{\n+\q}  &\text{if}\quad u_1=\frac1{\n!}(\1\otimes\partial)^\n,
\\ \\  \z^\beta  &\text{if}\quad u_1=\z^\beta \otimes D^{(\q)}.
    \end{array}
    \right.
\end{align*}
\end{itemize}
The proof is complete.
\end{proof}

We finally define the linear map $\Gamma:\envU(L)^*\otimes\A\to\A$ by
\begin{equation*}
\Gamma(f\otimes a):=\sum_{\gamma\in\M}\la \rhobar(f)(\z^\gamma),a\ra\z^\gamma+
\la f^{(\0)},a\ra \1, \qquad a\in\A,
\end{equation*}
with respect to the pairing  \eqref{eq:duality} between $\overline\A$ and $\A$,
where the sum is finite by Proposition \ref{pr:finite}.
We also use the notation $\Gamma_f:\A\to\A$ for $f\in\envU(L)^*$
\[
\Gamma_f (\z^\beta) :=\Gamma(f\otimes \z^\beta)= \sum_{\gamma\in\M}\la \rhobar(f)(\z^\gamma),\z^\beta\ra\z^\gamma+
\la f^{(\0)},\z^\beta\ra \1, \qquad \beta\in\M.
\]
In other words we have
\[
\Gamma_f = (\Lambda(f\otimes \cdot))^*, 
\]
in the pairing \eqref{eq:duality}.By Proposition \ref{prop:module}, we have that $\left(\A,\Gamma\right)$ is a \emph{right} $(\envU(L)^*,\trbar)$-module, namely for all $f_1,f_2\in\envU(L)^*$ and $a\in\A$, we have
\[
\Gamma\left( (f_1\trbar  f_2)\otimes a\right)=\Gamma\left( f_2\otimes \Gamma(f_1\otimes a)\right).
\]
In particular we obtain that for all $f_1,f_2\in\envU(L)^*$
\[
\Gamma_{f_1\trbar  f_2} = \Gamma_{f_2}\circ \Gamma_{f_1}.
\]
\bigskip

\printbibliography

\end{document}